\def\year{2022}\relax
\documentclass[letterpaper]{article} % DO NOT CHANGE THIS
\usepackage{aaai22}  % DO NOT CHANGE THIS
\usepackage{times}  % DO NOT CHANGE THIS
\usepackage{helvet}  % DO NOT CHANGE THIS
\usepackage{courier}  % DO NOT CHANGE THIS
\usepackage[hyphens]{url}  % DO NOT CHANGE THIS
\usepackage{graphicx} % DO NOT CHANGE THIS
\usepackage{natbib}  % DO NOT CHANGE THIS AND DO NOT ADD ANY OPTIONS TO IT
\usepackage{caption} % DO NOT CHANGE THIS AND DO NOT ADD ANY OPTIONS TO IT
\DeclareCaptionStyle{ruled}{labelfont=normalfont,labelsep=colon,strut=off} % DO NOT CHANGE THIS
\usepackage{algorithm}
\usepackage{algorithmic}
\usepackage{booktabs}

\usepackage{newfloat}
\usepackage{listings}
\floatstyle{ruled}
\newfloat{listing}{tb}{lst}{}
\floatname{listing}{Listing}
 \nocopyright

\author{
    Markus Brill,\textsuperscript{\rm 1}
    Th\'eo Delemazure,\textsuperscript{\rm 2} 
    Anne-Marie George,\textsuperscript{\rm 3}\\ 
    Martin Lackner,\textsuperscript{\rm 4}
    Ulrike Schmidt-Kraepelin\textsuperscript{\rm 1}
}
\affiliations {
    \textsuperscript{\rm 1} Research Group Efficient Algorithms, TU Berlin, Germany
    \textsuperscript{\rm 2} LAMSADE, Universit\'e Paris-Dauphine, France \\
    \textsuperscript{\rm 3} Department of Informatics, University of Oslo, Norway
    \textsuperscript{\rm 4} Databases and Artificial Intelligence Group, TU Wien, Austria
}

\title{Liquid Democracy with Ranked Delegations}

\usepackage{xspace}
\usepackage{amsmath}
\usepackage{amsthm}
\newcommand{\appSymb}{$\bigstar$}
\usepackage{thm-restate}
\usepackage{amssymb}
\usepackage{enumitem}
\usepackage{mathrsfs}
\usepackage{caption}
\usepackage{subcaption}

\usepackage[capitalise]{cleveref}	

\usepackage{tikz}
\usepackage{pgfplots}
\pgfplotsset{compat=1.15,
legend image code/.code={
\draw[mark repeat=2,mark phase=2]
plot coordinates {
(0cm,0cm)
(0.15cm,0cm)        %% default is (0.3cm,0cm)
(0.3cm,0cm)         %% default is (0.6cm,0cm)
};%
}}
\usetikzlibrary{shapes,arrows, decorations.pathmorphing}

\tikzstyle{delnode}=[draw= black,circle, thick,inner sep = 5]
\tikzstyle{votnode}=[draw= black, rectangle,inner sep = 7,thick]
\tikzstyle{isonode}=[draw= black, shape=regular polygon, regular polygon sides=3, rotate=0, minimum size=5pt,inner sep=4pt,thick]
\tikzstyle{rank1}=[]
\tikzstyle{rank2}=[]
\tikzstyle{rank3}=[dash pattern=on 0.5pt off 5pt, line cap=round]

\newcommand{\del}{\ensuremath{D}}
\newcommand{\cast}{\ensuremath{C}}

\newcommand{\pathv}{\ensuremath{\mathcal{P}_v}}
\newcommand{\seqv}{\ensuremath{\mathcal{S}_v}}

\newcommand{\fpre}{\ensuremath{\triangleright}\xspace}
\newcommand{\dfd}{DFD\xspace}
\newcommand{\bfd}{BFD\xspace}
\newcommand{\minsumseq}{MinSum\xspace}
\newcommand{\diffusion}{Diffusion\xspace}
\newcommand{\difforder}{\ensuremath{\triangleright_{\text{diff}}}}
\newcommand{\lex}{\ensuremath{\triangleright_{\text{lex}}}\xspace}
\newcommand{\sequences}{\ensuremath{\mathscr{S}}}
\newcommand{\lexrank}{Leximax\xspace}
\newcommand{\lrorder}{\ensuremath{\triangleright_{\sigma}}}
\newcommand{\minsumarb}{Borda\-Branching\xspace}
\newcommand{\minsumarbshort}{BordaBr.\xspace}

\newcommand{\guru}{representative\xspace}

\DeclareMathOperator*{\argmin}{arg\,min}
\DeclareMathOperator*{\argmax}{arg\,max}

\newtheorem{theorem}{Theorem}
\newtheorem{defn}[theorem]{Definition}
\newtheorem*{defn*}{Definition}
\newtheorem{lemma}[theorem]{Lemma}
\newtheorem{proposition}[theorem]{Proposition}

\newtheorem{corollary}[theorem]{Corollary}

\newtheorem*{claim}{Claim}

\newenvironment{claimproof}{%
	\proof}{\endproof}

\usepackage{etoolbox}
\usepackage{marginnote}
\usepackage[textsize=scriptsize]{todonotes} %,disable
\newcommand{\mytodo}[2]{\todo[size=\tiny, color=#1!50!white]{#2}\xspace}
\newcommand{\myrevtodo}[2]{{%
		\let\marginpar\marginnote
		\reversemarginpar
		\renewcommand{\baselinestretch}{0.8}%
		\todo[size=\tiny, color=#1!50!white]{#2}\xspace}}
\newcommand{\myinlinetodo}[2]{\todo[size=\small, color=#1!50!white, inline, caption={}]{#2}\xspace}
\newcommand{\registerAuthor}[3]{%
	\expandafter\newcommand\csname #2com\endcsname[1]{\mytodo{#3}{\textsc{#2}: 
	##1}}%
	\expandafter\newcommand\csname 
	#2revcom\endcsname[1]{\myrevtodo{#3}{\textsc{#2}: ##1}}%
	\expandafter\newcommand\csname 
	#2inline\endcsname[1]{\myinlinetodo{#3}{\textsc{#2}: ##1}}%
	\expandafter\newcommand\csname 
	#2inlineLater\endcsname[1]{\lv{\myinlinetodo{#3}{\textsc{#2}: ##1}}}%
}

\usepackage{xcolor}
\definecolor{color2}{RGB}{1,113,0}
\definecolor{maincolor}{RGB}{149,29,19}
\usepackage{pifont}
\newcommand{\cmark}{\ding{51}}
\newcommand{\xmark}{\ding{55}}
\newcommand{\yes}{\textcolor{color2}{\cmark}}
\newcommand{\no}{\textcolor{maincolor}{\xmark}}

\begin{document}
\maketitle

\begin{abstract}
Liquid democracy is a novel paradigm for collective decision-making that gives agents the choice between casting a direct vote or delegating their vote to another agent. We consider a generalization of the standard liquid democracy setting by allowing agents to specify multiple potential delegates, together with a preference ranking among them. This generalization increases the number of possible delegation paths and enables higher participation rates because fewer votes are lost due to delegation cycles or abstaining agents. In order to implement this generalization of liquid democracy, we need to find a principled way of choosing between multiple delegation paths. In this paper, we provide a thorough axiomatic analysis of the space of \textit{delegation rules}, i.e., functions assigning a feasible delegation path to each delegating agent. In particular, we prove axiomatic characterizations as well as an impossibility result for delegation rules. We also analyze requirements on delegation rules that have been suggested by practitioners, and introduce novel rules with attractive properties. By performing an extensive experimental analysis on synthetic as well as real-world data, we compare delegation rules with respect to several quantitative criteria relating to the chosen paths and the resulting distribution of voting power. Our experiments reveal that delegation rules can be aligned on a spectrum reflecting an inherent trade-off between competing objectives. 
\end{abstract}

\section{Introduction}

Liquid democracy is a novel decision-making paradigm that aims to reconcile the idealistic appeal of {direct democracy} with the practicality of {representative democracy} by giving agents a choice regarding their mode of participation: For every given issue, agents can choose whether they want to vote directly or whether they want to \textit{delegate} their vote to another agent.
The mode of participation can differ for different issues and the choice of mode (including the choice of whom to delegate to) can be altered at any time. This enables a flexible and dynamic scheme of representation on a per-issue basis \citep{blum16liquid}.

An important principle of liquid democracy is the transitivity of delegation (sometimes called \textit{metadelegation}): if $A$ delegates to $B$ and $B$ delegates to $C$, then $C$ has a total voting weight of $3$ (her own vote plus those of $A$ and~$B$). More generally, delegated votes travel along \textit{delegation paths} until they reach an agent who votes directly. This agent at the endpoint of the path (sometimes referred to as ``guru'') then has a voting weight corresponding to all the agents whose delegation paths end up with him or her.   

Introducing the option of delegating one's vote arguably lowers the bar for participation, as it does not require agents to get informed about an issue before making their vote count \citep{Ford02a,BBCV11a,Vals21a}. This is particularly true if delegations can be declared globally and/or for whole subject areas.\footnote{For instance, this is possible in the open-source software \textit{LiquidFeedback} \citep{BKNS14a}, where more fine-grained delegations override global delegations, and direct voting on any particular issue overrides all delegations.}
As a result, liquid democracy has the potential to achieve substantially higher rates of (direct and indirect) participation than direct voting. 

However, vote delegations do not always work out as intended. 
For instance, if a delegation path ends in an agent who \textit{abstains} (i.e., neither delegates nor votes directly), the vote is inevitably lost. 
Moreover, if a group of agents  
delegate to each other in a cyclic fashion and no group member votes directly, they form a ``delegation cycle'' and all their votes, together with the votes of agents whose delegation paths lead to the cycle, are lost.
Finally, delegation paths can potentially get very long, calling into question the degree of trust between the agents at the beginning and end of the path.

To address these issues, we consider an extension of liquid democracy that lets agents declare \textit{multiple} potential delegates, together with a ranking among them specifying the agent's delegation preferences: By ranking potential delegate $X$ higher than potential delegate $Y$, an agent indicates that she would prefer delegating her vote to $X$, but would also be happy with a delegation to $Y$ in the case that the delegation to $X$ leads to one of the problems mentioned in the previous paragraph. 
This functionality was implemented, e.g., in a liquid democracy experiment at \textit{Google} \citep{HaLo15a}.
Allowing agents to declare \textit{ranked delegations} enlarges the space of possible delegation paths and thereby increases the likelihood of ``successful'' delegation paths. In the case that multiple such paths exist, however, we need a principled way to choose among them. This is accomplished by so-called \emph{delegation rules}, which select delegation paths based on the delegation preferences of agents. 

\paragraph{Our Contribution}

We propose a graph-theoretic model capturing the ranked delegation setting and study the space of delegation rules. 
We define \textit{sequence rules} as the subclass of delegation rules that can be rationalized by an ordering over delegation paths that only takes the ranks of edges
into account.
For most delegation rules considered in the literature, we establish their membership in this class by uncovering their respective order. Our systematic approach leads us to introduce novel sequence rules with attractive properties. 

In our axiomatic analysis, we generalize a result by \citet{KoRi20a}, showing that all but one rule studied in our paper satisfy \emph{guru-participation}. We also study the axiom \emph{copy-robustness}, which is motivated by practical considerations. Notwithstanding a strong impossibility result for the subclass of sequence rules, we construct a copy-robust delegation rule by taking a ``global'' approach. Lastly, we give axiomatic characterizations for two sequence rules.

We complement our theoretical results with an extensive experimental analysis using synthetic and real-world data. We empirically compare delegation rules with respect to several quantitative criteria such as the length of chosen paths or the distribution of voting power that results from those paths. Interestingly, we find that the delegation rules form a spectrum, reflecting an inherent trade-off between short paths and paths containing top-ranked delegations.

\paragraph{Related Work}\label{sec:related-work}

As illustrated by \citet{Behr17a}, some of the ideas underlying liquid democracy date back to \citet{Dodg84a}, \citet{Tull67a}, and \citet{miller1969program}.  
Liquid democracy has been studied theoretically, and applied practically, in various ways in recent years \citep{Paulin}. 

Liquid democracy settings with ranked delegations have been considered by several authors. \citet{BeSw15a} 
illustrate the incompatibility of a set of axioms,
several of which we discuss in \Cref{sec:axiomatic}. 
\citet{KoRi20a} focus on two simple delegation rules and on participation axioms. Their observation that \textit{breadth-first delegation} satisfies 
guru-participation also follows from our more general result in \Cref{sec:axiomatic}.
\citet{KKM+21a} 
establish a connection between ranked delegations and branchings in directed graphs, which we build upon in \Cref{sec:branching}, and focus on the computational complexity of finding ``popular'' branchings. 
\citet{CGN20a} propose a generalisation of the ranked delegations setting where agents can express complex delegations involving \textit{sets} of delegates on each level of the ranking. When restricted to our setting, their ``unravelling procedures'' reduce to variants of the rule we call \textit{Diffusion}.
Both Kavitha et al. and Colley et al. assume that agents specify \textit{backup votes}, i.e., votes that are used when no delegation path exists. In contrast, our model does not necessitate this rather demanding assumption. 

Other settings with multiple delegations have been considered as well. 
\citet{GKMP18b} let agents specify multiple delegation options, but without specifying preferences among them. Their objective is to assign delegations in such a way as to minimize the maximal voting weight of agents.  
\citet{KMP21a} consider an epistemic setting in which each agent has a competence level (i.e., probability of making the ``correct'' voting decision) and an approved subset of other agents. They are interested in (possibly randomized) delegation mechanisms that increase the likelihood of a correct decision compared to direct voting. Since delegation is only allowed to more competent agents, there can be no delegation cycles in their model. 
Finally, \citet{ChGr17a} and \citet{BrTa18a} let agents delegate different decisions to different delegates and explore ways to ensure individual rationality.

\section{The Model} \label{sec:model}

We start our exploration of the ranked delegation setting by presenting a graph-theoretic model. 
We focus on a single issue that is to be decided upon and assume that we are given 
(1) the set of agents that are casting a direct vote (\textit{casting} voters) as well as 
(2) for each non-casting agent, a (possibly empty) set of other agents 
together with a ranking over them representing delegation preferences.\footnote{An empty set of potential delegates corresponds to abstaining.} 
Based on this information, we want to assign non-casting agents to casting voters by choosing delegation paths.  
Our model focuses on determining the voting weights of casting voters and is, therefore, independent of any particular method for aggregating the votes of the casting voters. Separating the delegation mechanism from the voting method allows us to analyze the former in isolation, and also simplifies the model.\footnote{As we will see in \Cref{sec:axiomatic}, axioms formulated in more complex models \citep{BeSw15a,KoRi20a,CGN20a} can be translated into ours.}

We represent a ranked delegation instance as a pair \linebreak $(G,r)$, where $G=(C \cup D \cup I 
,E)$ is a directed graph. The set of vertices of $G$ corresponds to the set of all voters (agents) and is partitioned into three sets $C$, $D$, and $I$ with the following properties:
\begin{itemize} 
    \item nodes in $C$ have no outgoing edges in $G$ (we refer to voters in $C$ as \emph{casting voters});
    \item for each node in $D$, there exists at least one path in $G$ to a voter in $C$ (we refer to voters in $D$ as \emph{delegating voters}); % and 
    \item for each node in $I$, there exists no path in $G$ leading to a voter in $C$ (we refer to voters in $I$ as \emph{isolated voters}). 
\end{itemize}
For the set of all voters we write $V = C \cup D \cup I$. Note that for any graph $G=(V,E)$ with a distinguished set $C \subseteq V$ of casting voters, the sets $D$ and $I$ are uniquely defined. 
The second element of the instance, $r$, is a \emph{rank function} on the set of edges that encodes, for every node, a linear order over its set of outgoing edges. Formally, $r: E \rightarrow \mathbb{N}_{\geq 1}$ is a function such that 
$\{r(e) \mid e \in \delta_G^+(v)\} = \{1, \dots, |\delta^+_G(v)|\}$ for all $v \in V \setminus C$, 
where $\delta^+_G(v)$ is the set of outgoing edges of $v$ in~$G$. If $r((v,x))<r((v,y))$ this is interpreted as voter $v$ preferring to delegate to voter $x$ over delegating to voter $y$. 
For an example of a ranked delegation instance, see \Cref{fig:running-example}.

For a non-casting voter $v \in V \setminus \cast$, we denote the set of all simple paths from $v$ to some casting voter by $$\pathv:=\{P \mid P \text{ is a simple } v\text{-}w\text{-path and } w \in \cast\}.$$ Observe that $\mathcal{P}_v$ is empty if and only if $v \in I$. This is in particular the case when $\delta^+_G(v) = \emptyset$ (in which case we call $v$ an \emph{abstaining voter}), but it also happens if $v$ has only outgoing edges towards other isolated voters. 
In \Cref{fig:running-example}, $\mathcal{P}_g = \emptyset$ and 
$\mathcal{P}_d=\{(d,j),(d,e,b,c,i),(d,e,f,k)\}$. \footnote{We interpret paths as \emph{edge} sequences. However, for brevity, we state \emph{node} sequences whenever referring to paths in \Cref{fig:running-example}. {For convenience, we interpret paths as sets of edges in  \Cref{def:confluence}.}}

Our task is now to define sensible rules which assign delegating voters to casting ones (via a path in $G$). 

\begin{defn}
A \emph{delegation rule} $f$ is a function that takes as input a ranked delegation instance $(G,r)$ and a delegating voter $v \in D$ and outputs a path from $v$ to a casting voter, i.e., $f(G,r,v) \in \pathv$ for all $v \in \del$. We also write $f(v)$, whenever the instance $(G,r)$ is clear from the context.  
\end{defn}

Since isolated voters cannot be contained in any of the paths in $\pathv$, their influence on the problem is limited. Thus, it will often be convenient to omit the isolated voters from the instance. However, we argue that the information about their existence should still be reflected within the rank function $r$.%
\footnote{E.g., consider the edge from agent $f$ to agent $k$ in \cref{fig:running-example}. 
If we were to delete all isolated agents and adjust the rank function, the new rank of this edge would be $2$, even though $k$ is only $f$'s fourth choice. Such examples can be made arbitrarily ``extreme.''}
More formally, we define the graph $\bar{G} = (\bar{V},\bar{E})$ with $\bar{V} = C \cup D$ and $\bar{E} = \{(u,v) \in E\mid u,v \in \bar{V}\}$, and 
$\bar{r}(e) = r(e)$ for all $e \in \bar{E}$ 
and refer to $(\bar{G},\bar{r})$ as the \emph{reduced instance} of $(G,r)$.

\newcommand{\first}{node [circle, fill=white, inner sep=0.05cm, pos=0.4] {$1$}}
\newcommand{\second}{node [circle,fill=white, inner sep=0.05cm, pos=0.4] {$2$}}
\newcommand{\third}{node [circle,fill=white, inner sep=0.05cm, pos=0.4] {$3$}}
\newcommand{\fourth}{node [circle,fill=white, inner sep=0.05cm, pos=0.4] {$4$}}

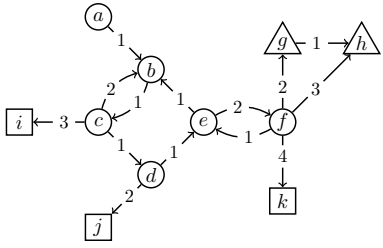
\begin{figure}[t]
    \centering
    \scalebox{.7}{
    \begin{tikzpicture}
 \node[delnode,label=center:\textcolor{black}{\large ${b}$},minimum size=5](2) at (0,1) {};
 \node[delnode,label=center:\textcolor{black}{\large ${d}$}, minimum size=5](5) at (0,-1) {};
  \node[delnode,label=center:\textcolor{black}{\large ${e}$}, minimum size=5](7) at (1,0) {};
   \node[delnode,minimum size=5,label=center:\textcolor{black}{\large ${c}$}](3) at (-1,0) {};
    \node[delnode,minimum size=5,label=center:\textcolor{black}{\large ${a}$}](1) at (-1,2) {};
 \node[votnode,minimum size=8,label=center:\textcolor{black}{\large ${j}$}](6) at (-1,-2) {};
  \node[delnode,minimum size=5,label=center:\textcolor{black}{\large ${f}$}](8) at (2.5,0) {};
   \node[votnode,minimum size=8,label=center:\textcolor{black}{\large ${i}$}](4) at (-2.5,0) {};
  \node[isonode,minimum size=5,label=center:\textcolor{black}{\large ${g}$}](9) at (2.5,1.5) {};
    \node[votnode,minimum size=8,label=center:\textcolor{black}{\large ${k}$}](10) at (2.5,-1.5) {};
        \node[isonode,minimum size=8,label=center:\textcolor{black}{\large ${h}$}](11) at (4,1.5) {};
    
    \draw[thick,->] (1) -- (2) \first;
    \draw[thick,->, bend left] (2) to node [circle, fill=white, inner sep=0.05cm, pos=0.4] {$1$} (3);
    \draw[thick,->] (3) -- (5) \first;
    \draw[thick,->] (5) -- (7) \first;
    \draw[thick,->] (7) -- (2) \first;
     \draw[thick,->,bend left] (8) to node [circle, fill=white, inner sep=0.05cm, pos=0.4] {$1$} (7);
     
     \draw[thick,->, bend left] (3) to node [circle, fill=white, inner sep=0.05cm, pos=0.4] {$2$}(2);
     \draw[thick,->] (5) -- (6)\second;
      \draw[thick,->, bend left] (7) to node [circle, fill=white, inner sep=0.05cm, pos=0.4] {$2$}(8);
       \draw[thick,->] (8) -- (9)\second;
       
        \draw[thick,->] (8) -- (10)\fourth;
         \draw[thick,->] (3) -- (4)\third;

        \draw[thick,->] (9) -- (11)\first;
        \draw[thick,->] (8) -- (11)\third;

\end{tikzpicture}}\caption{Example of a ranked delegation instance $(G,r)$. Casting voters are indicated by squares, 
    delegating voters by circles and isolated voters by triangles.}\label{fig:running-example}
\end{figure}

\section{Delegation Rules} \label{sec:rules}

Before introducing and analyzing several delegation rules, we formalize two fundamental concepts that will help us to classify delegation rules: \emph{confluence} and \emph{sequence rules}. 

To motivate the former, observe that the definition of a delegation rule does not forbid that selected paths cross. Consider for example the instance in Figure \ref{fig:running-example}: There exists a delegation rule assigning path $(c,d,e,f,\dots)$ to agent $c$ but path $(d,e,b,\dots)$ to agent $d$. 
This can be seen as counter-intuitive as it is incompatible with the view that voter $e$ takes the decision on how to delegate incoming delegations based on its delegation preferences.
Also, it can be difficult for voter $e$ to keep track and evaluate the decisions made by all casting voters who have been reached via delegation paths passing through $e$. In this context, \citet{GKMP18b} argue that a single representative per delegating voter is needed ``to preserve the high level of accountability guaranteed by classical liquid democracy.'' 
We define a natural subclass of delegation rules in which this is guaranteed. In reference to \emph{confluent flows} in networks \cite{CRS06a}, we call such delegation rules \emph{confluent}. 

\begin{defn}\label{def:confluence}
A delegation rule $f$ is \emph{confluent} if, for every instance $(G,r)$, every delegating voter $v \in D$ has outdegree exactly one in $(V,\bigcup_{v \in D} f(G, r, v))$. 
\end{defn}

Besides confluence, two natural---and often conflicting---objectives are minimizing the length of paths and minimizing the ranks of edges contained in paths. Short paths are particularly motivated by the fact that delegations are a form of trust propagation and it is debatable to what extent trust is transitive. At the same time, votes should be delegated to highly trusted agents, motivating the selection of paths with top-ranked edges.
Both of these objectives can be evaluated by only considering the sequence of ranks appearing along a path. Indeed, most delegation rules introduced in the literature can be expressed as choosing, among all available delegation paths, the one with the ``best'' rank sequence. We formalize this subclass of delegation rules as \textit{sequence rules}. 

To do so, we need some notation. 
Let $\sequences$ be the set of all finite sequences of numbers in $\mathbb{N}_{\geq 1}$.\footnote{For technical reasons, $\sequences$ also includes the empty sequence $()$ with length $0$ and (by convention) maximum rank $0$.}
We define the \emph{sequence} of a path $P=(e_1, e_2, \dots, e_{\ell})$ as $s(P):=(r(e_1), \dots, r(e_{\ell}))$ and denote the set of all sequences from $v$ to some casting voter by $\seqv := \{s(P) \mid P \in \pathv\}$. 
For a sequence $s$, we write $s_i$ to refer to the $i$-th element of $s$, $|s|$ for the length of $s$, $\max(s)$ for the maximal entry, and use the notation $(s,x)$ to denote $s$ extended by a number $x \in \mathbb{N}_{\geq 1}$ or a sequence $x$.

We call a delegation rule a sequence rule if it, explicitly or implicitly, defines a relation $\fpre$ over $\mathscr{S}$ and, when confronted with $\pathv$, guarantees a unique maximum element of $\seqv$ w.r.t.~$\fpre$ and selects the corresponding path. 
It is clearly sufficient to define $\fpre$ on sets of \emph{comparable} sequences:

\begin{defn}
Two sequences $s, s' \in \mathscr{S}$ are said to be comparable if there exists an instance $(G,r)$ and a vertex $v \in V$ such that $\mathcal{S}_v = \{s,s'\}$. A set $\mathcal{S} \subseteq \sequences$ is said to be comparable if all elements are pairwise comparable.
\end{defn} 

Not all pairs of sequences are comparable; e.g. , there is no instance with $\mathcal{S}_v = \{(1),(1,2)\}$ for some $v$, as the first elements of $(1)$ and $(1,2)$ would correspond to the same edge~$e$ (of rank $1$) that is outgoing from $v$. Thus, the head of $e$ is a casting voter (due to $(1) \in \mathcal{S}_v$) but has an outgoing edge (due to $(1,2) \in \mathcal{S}_v$), a contradiction. 
This observation can be extended to any situation in which $s$ is a prefix of $s'$ (i.e. $|s| < |s'|$ and $s_i = s'_i$ for all $i \in \{1, \dots, |s|\}$).
Proofs for results marked by (\appSymb) can be found in the appendix.

\begin{restatable}[\appSymb]{proposition}{obsCompIff}\label{pr: obsCompIff}
Two distinct non-empty sequences $s,s' \in \mathscr{S}$ are comparable iff none is a prefix of the other. 
\end{restatable}

\newcommand{\obsCompIffProof}{
\begin{proof}
First, let $s,s' \in \mathscr{S}$ and assume wlog that $s$ is a prefix of $s'$. Assume for contradiction that there exists an instance $(G,r)$ and a voter $v \in V$ with $P,P' \in \mathcal{P}_v$ such that $s(P)=s$, $s(P')=s'$ and $\{s,s'\} = \mathcal{S}_v$. As $r((v,w))\neq r((v,w'))$ for all $(v,w),(v,w')\in E$ with $w\neq w'$, the first $|s|$ edges of $P$ and $P'$ coincide. Let $w$ be the voter at the end of $P$. As $P \in \mathcal{P}_v$, we know that $w \in C$. However, as $|P'| > |P|$, $w$ needs to have an outgoing edge, a contradiction because casting voters have no outgoing edges. 

For the other direction, consider any $s,s' \in \mathscr{S}$ such that none is a prefix of the other. Let 
$i < \min(|s|,|s'|)$ be the first position at which the sequences differ, i.e., $s_{i} \neq s'_{i}$ and $s_j = s'_j$ for all $j < i$. We will create an instance $(G,r)$, which is later extended by abstaining voters, as follows: We start by introducing four nodes $v,w,x,x'$, 
a $v$-$x$-path $P$ of length $|s|$ that goes via $w$ and a $v$-$x'$-path $P'$ of length $|s'|$ that also goes via $w$. In particular, the two paths $P$ and $P'$ coincide up until node $w$ and then split. Figure~\ref{sfig: reduced graph} shows the two paths with their edge ranks, where delegating voters are depicted as circles and casting voters as squares. 
Lastly, in order to guarantee that the ranks are fitting, i.e., $s(P)=s$ and $s(P')=s'$, for every node with outgoing edge with rank $s_i$, respectively $s'_i$, we introduce $s_i -1$, respectively $s'_i -1$, many edges to abstaining voters. One exception is node $w$ for whom we introduce $\max(s_i,s'_i) - 2$ many edges. See Figure~\ref{sfig: full graph} for an example construction, where abstaining voters are depicted as triangles. Note that our construction is so that $\mathcal{S}_v = \{s,s'\}$. 
\end{proof}
}

\newif\ifobsCompIffAppendix 
\obsCompIffAppendixtrue  

\ifobsCompIffAppendix
\else
\obsCompIffProof
\fi

We are now ready to define the class of sequence rules. 

\begin{defn} \label{def:sequenceRules}
A delegation rule $f$ is a \emph{sequence rule} if there exists a relation $\fpre$ over $\mathscr{S}$, which, if restricted to any comparable subset of $\mathscr{S}$, is a linear order and for any instance $(G,r)$ and $v \in V$ it holds that $s(f(v)) = max_{\fpre} \{\seqv\}$.
\end{defn}

If a relation $\fpre$ proves that $f$ is a sequence rule, we also say that $\fpre$ \emph{induces} $f$.

\subsection{Basic Sequence Rules}

In the following we describe rules which are induced by natural linear orders over the entire set $\mathscr{S}$. The first two rules have also been considered by 
\citet{KoRi20a}. Let $\lex$ be the \textit{lexicographic order} over~$\sequences$. That is, for two distinct sequences $s,s' \in \sequences$, let $i$ be the smallest index such that $s_i \neq s'_i$. Then, $s \;\lex s'$ iff $s_i < s'_i$. If no such index exists, then $s \lex s'$ iff $|s|<|s'|$.

\smallskip
\noindent\textbf{Depth-first delegation (DFD):} The sequence rule induced by the lexicographic order $\lex$ over $\mathscr{S}$.

\smallskip 
\noindent\textbf{Breadth-first delegation (BFD):} The sequence rule induced by the linear order over $\sequences$ that prefers shorter sequences over longer ones and uses $\lex$ for tie-breaking.

\smallskip

While \dfd focuses on the objective of selecting delegations with top ranks, \bfd primarily selects short paths. 
As a consequence, each of them have at least one obvious drawback: \dfd is mostly indifferent about the length of a path and may select very long paths, whereas \bfd mostly ignores the ranks of a path and hence can select paths containing edges that are ranked extremely low. 

We propose a third, natural delegation rule, which arguably strikes a better balance between the two objectives of minimizing the length and the ranks of a path.

\smallskip
\noindent\textbf{\minsumseq:} The sequence rule induced by the linear order over $\sequences$ that orders sequences by their sum of ranks and uses the lexicographic order for tie-breaking.
\smallskip

To illustrate these rules, consider again the example instance in \Cref{fig:running-example}. For delegating voter $a$, 
\dfd selects the path $(a,b,c,d,e,f,k)$ (with sequence $(1,1,1,1,2,4)$), 
\bfd selects the path $(a,b,c,i)$ (with sequence $(1,1,3)$), and
\minsumseq selects $(a,b,c,d,j)$ (with sequence $(1,1,1,2)$).

Next, we characterize confluent sequence rules.

\begin{restatable}[\appSymb]{theorem}{prConfluent}\label{pr: prefix-invariant}
A sequence rule induced by $\fpre$ is confluent if and only if for all $s \in \sequences$ we have (i) if $s$ is comparable to some $s' \in \sequences$ and $x \in \mathbb{N}_{\geq 1}$, then $s \fpre s' \Leftrightarrow (x, s) \fpre (x, s')$ and (ii) if $s' \in \sequences$ and $s$ is comparable to $(s', s)$, then $s \fpre (s', s)$.
\end{restatable}

\newcommand{\prConfluentProof}{
\begin{proof}
First, we show that any confluent sequence rule satisfies the two properties (i) and (ii). To show (i), consider two comparable sequences $s$ and $s'$ and construct a reduced instance as shown in \cref{sfig: inst1}. For details on how to construct a feasible instance from such a graph, consider the proof of \cref{pr: obsCompIff}. In the created instance it holds that $\mathcal{S}_v = \{s, s'\}$ and $\mathcal{S}_u = \{(x, s), (x, s')\}$. In particular, any path from $u$ to a casting voter involves $v$. Assume $s \fpre s'$, i.e., $v$'s vote is assigned to casting voter $c$. Then, by confluence $u$'s vote also has to be delegated to $c$, which implies $(x, s) \fpre (x, s')$. Similarly, assume $(x,s) \fpre (x,s')$. Then, $u$'s vote is assigned to casting voter $c$. By confluence, $v$'s vote needs to be assigned to casting voter $c$ as well, hence $s \fpre s'$. 

To show (ii), consider $s,s' \in \sequences$ such that $s$ and $(s',s)$ are comparable and construct a graph as shown in \cref{sfig: inst2}. By using dummy voters, this graph can be extended to a feasible instance of our problem. In this instance, $\mathcal{S}_v = \mathcal{S}_u = \{(s', s), (s)\}$. If $(s', s) \fpre s$, then $c$ is the assigned casting voter of $u$ and $c'$ the assigned casting voter of $v$. However, this is a contradiction to confluence because $v$ is on the path from $u$ to $c$ and thus $u$ must have the same assigned casting voter as $v$. Thus $s \fpre (s', s)$, which proves property (ii).

Now, consider a sequence rule induced by a relation $\fpre$ that satisfies properties (i) and (ii), and assume that there exists an instance in which this rule is not confluent. That is, there exists an instance with three nodes $u,u'$ and $v$, such that $(v,w)$ is on the chosen delegation path $P$ for $u$, while the chosen delegation path $P'$ for $u'$ uses some edge $(v,w')$ with $w' \neq w$. Observe that we can assume without loss of generality that $u' = v$, since the first edge of the path for $v$ uses an edge that is different to at least one of $(v,w)$ and $(v,w')$. Hence, we assume $u'=v$ in the following. We distinguish two cases.

\textbf{Case 1:} $P'$ does not pass over any node contained in the path $P$ from $u$ to $v$. Let $s = s(P) = (s^{(1)}, s^{(2)})$ be the sequence of ranks of $P$, such that $s^{(1)}$ is the sequence for the subpath from $u$ to $v$ and $s^{(2)}$ the sequence for the subpath from $v$ to $c$. Let $s' = s(P')$ be the sequence of ranks of $P'$. Consider \cref{sfig: inst4} for a sketch of the situation. Note that the two paths can branch and merge multiple times and in particular $c$ and $c'$ can be equal. However, by the case assumption we have that $(s^{(1)}, s') \in \mathcal{S}_u$. Because the chosen path for $v$ is $P'$, we know that $s' \fpre s^{(2)}$. Property (i) (applied multiple times) then implies that $(s^{(1)}, s') \fpre (s^{(1)}, s^{(2)}) = s$ which is a contradiction to $P$ being chosen for $u$.

\textbf{Case 2:} $P'$ contains some node from $P$ within the subpath from $u$ to $v$. Let $w$ be the first node on the path $P$ from $u$ to $v$ that $P'$ contains. Then, let $s=(s^{(1)},s^{(2)},s^{(3)})$ be such that $s^{(1)}$ corresponds to the sequence induced by $P$ from $u$ to $w$, $s^{(2)}$ to the sequence induced by $P$ from $w$ to $v$, and $s^{(3)}$ to the sequence induced by $P$ from $v$ to $c$. In particular, if $u=w$ then $s^{(1)}$ is empty. 
Similarly, we denote $s' = s(P') = (s'^{(1)}, s'^{(2)})$, where $s'^{(1)}$ is the sequence of the subpath of $P'$ from $v$ to $w$ and $s'^{(2)}$ the sequence from $w$ to $c'$. Consider \cref{sfig: inst5} for a sketch of the situation.
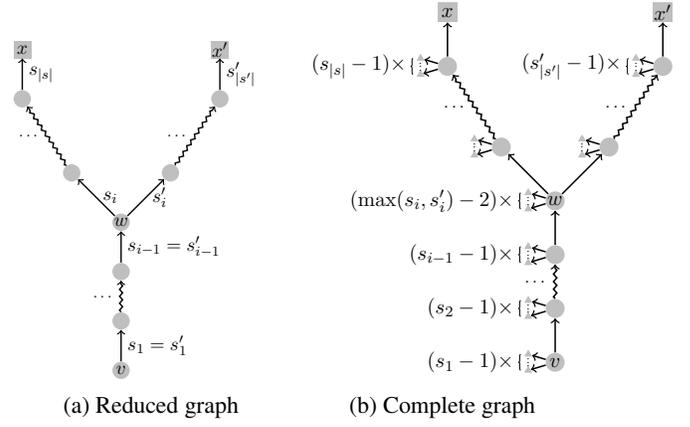
\begin{figure}[t!]
     \begin{subfigure}[b]{0.45\linewidth}
     \resizebox{!}{4.5cm}{
    \begin{tikzpicture}
        \tikzstyle{vertex}=[circle,fill=black!25,minimum size=10pt,inner sep=0pt]
        \tikzstyle{caster}=[rectangle,fill=black!25,minimum size=10pt,inner sep=0pt]
        \tikzstyle{edge} = [draw,thick,->]
        \tikzstyle{edgewavy} = [draw,thick,->,decorate, decoration={
                                    zigzag,
                                    segment length=4,
                                    amplitude=.9,post=lineto,
                                    post length=2pt
                                }]
        \tikzstyle{weight} = [font=\small]
        
        % \node[] at (-2, 2){$G$};
        
        \foreach \pos/\name/\lab in {{(0,1)/v/v}, {(0,2)/a/}, {(0,3)/b/}, {(0,4)/c/w}, 
                                {(-1,5)/s/}, {(-2,6.5)/t/}, 
                                {(1,5)/d/}, {(2,6.5)/e/}}
                            \node[vertex] (\name) at \pos{\large$\lab$};
        \foreach \pos/\name/\lab in {{(-2,7.5)/u/x}, {(2,7.5)/f/x'}}
                            \node[caster] (\name) at \pos{\large$\lab$};
        
        \foreach \source/ \dest /\weight in {a/b/, s/t/, d/e/}
                \path[edgewavy] (\source) -- node[weight, left] {$\dots$} (\dest);
                
        \foreach \source/ \dest /\weight in {v/a/s_1=s'_1, b/c/s_{i-1}=s'_{i-1}, c/s/s_{i}, c/d/s'_{i}, t/u/s_{|s|}, e/f/s'_{|s'|}}
                \path[edge] (\source) -- node[weight, right] {\large$\weight$} (\dest);
    \end{tikzpicture}
    }
    \caption{Reduced graph}\label{sfig: reduced graph}
    \end{subfigure}
    \begin{subfigure}[b]{0.45\linewidth}
     \resizebox{!}{5cm}{
    \begin{tikzpicture}
        \tikzstyle{vertex}=[circle,fill=black!25,minimum size=10pt,inner sep=0pt]
        \tikzstyle{caster}=[rectangle,fill=black!25,minimum size=10pt,inner sep=0pt]
        \tikzstyle{abstainer}=[shape=regular polygon, regular polygon sides=3, rotate=0, fill=black!25, minimum size=5pt,inner sep=0pt]
        \tikzstyle{edge} = [draw,thick,->]
        \tikzstyle{edgewavy} = [draw,thick,->,decorate, decoration={
                                    zigzag,
                                    segment length=4,
                                    amplitude=.9,post=lineto,
                                    post length=2pt
                                }]
        \tikzstyle{weight} = [font=\small]
        
        \foreach \pos/\name/\lab in {{(0,1)/v/v}, {(0,2)/a/}, {(0,3)/b/}, {(0,4)/c/w}, 
                                {(-1,5)/s/}, {(-2,6.5)/t/}, 
                                {(1,5)/d/}, {(2,6.5)/e/}}
                            \node[vertex] (\name) at \pos{\large$\lab$};
                            
        \foreach \pos/\name/\lab in {{(-2,7.5)/u/x}, {(2,7.5)/f/x'}}
                            \node[caster] (\name) at \pos{\large$\lab$};
                            
        \foreach \pos/\name/\lab in {{(-.5,.85)/v1/}, {(-.5,1.15)/v2/},
                                {(-.5,1.85)/a1/}, {(-.5,2.15)/a2/},
                                {(-.5,2.85)/b1/}, {(-.5,3.15)/b2/},
                                {(-.5,3.85)/c1/}, {(-.5,4.15)/c2/},
                                {(.5,4.85)/d1/}, {(.5,5.15)/d2/},
                                {(1.5,6.35)/e1/}, {(1.5,6.65)/e2/},
                                {(-1.5,4.85)/s1/}, {(-1.5,5.15)/s2/},
                                {(-2.5,6.35)/t1/}, {(-2.5,6.65)/t2/}}
                            \node[abstainer] (\name) at \pos{};
                            
        \newcommand\mydots{\ifmmode\ldots\else\makebox[1em][c]{.\hfil.\hfil.}\fi}
        \foreach \pos/\name/\lab in {{(-.5,1)/vdots/\rotatebox{90}{\tiny\mydots}},
                                {(-.5,2)/adots/\rotatebox{90}{\tiny\mydots}},
                                {(-.5,3)/bdots/\rotatebox{90}{\tiny\mydots}},
                                {(-.5,4)/cdots/\rotatebox{90}{\tiny\mydots}},
                                {(.5,5)/ddots/\rotatebox{90}{\tiny\mydots}},
                                {(1.5,6.5)/edots/\rotatebox{90}{\tiny\mydots}},
                                {(-1.5,5)/sdots/\rotatebox{90}{\tiny\mydots}},
                                {(-2.5,6.5)/tdots/\rotatebox{90}{\tiny\mydots}}}
                            \node (\name) at \pos{\large\lab};
                            
        \foreach \x/\y/\lab in {{-.5/1/(s_1 - 1) \times},
                                {-.5/2/(s_2 - 1) \times},
                                {-.5/3/(s_{i-1} - 1) \times},
                                {-.5/4/(\max(s_{i},s'_{i}) - 2) \times},
                                {1.5/6.5/(s'_{|s'|} - 1) \times},
                                {-2.5/6.5/(s_{|s|} - 1) \times}}
                            \draw [decorate,decoration={brace,amplitude=1pt},xshift=-4pt,yshift=0pt]
(\x,\y-0.15) -- (\x,\y+0.15) node [black,midway, left]{\large $\lab$};
        
        \foreach \source/ \dest /\weight in {a/b/, s/t/, d/e/}
                \path[edgewavy] (\source) -- node[weight, left] {$\dots$} (\dest);
                
        \foreach \source/ \dest /\weight in {v/a/r_1=r'_1, b/c/s_{i-1}=s'_{i-1}, c/s/r_{i}, c/d/r'_{i}, 
                                            t/u/r_{k}, e/f/r'_{\ell},   
                                            v/v1/, v/v2/,
                                            a/a1/, a/a2/,                      
                                            b/b1/, b/b2/,                       
                                            c/c1/, c/c2/,                       
                                            d/d1/, d/d2/,                      
                                            e/e1/, e/e2/,                      
                                            s/s1/, s/s2/,                      
                                            t/t1/, t/t2/}
                \path[edge] (\source) -- node[weight, right] {} (\dest);
    \end{tikzpicture}
    }
    \caption{Complete graph}\label{sfig: full graph}
    \end{subfigure}
    \caption{
    Example for the construction in the proof of Proposition~\ref{pr: obsCompIff}.
    Casting voters are indicated by rectangles
    and delegating voters are indicated by circles.
    } \label{fig: obsCompIff}
\end{figure}
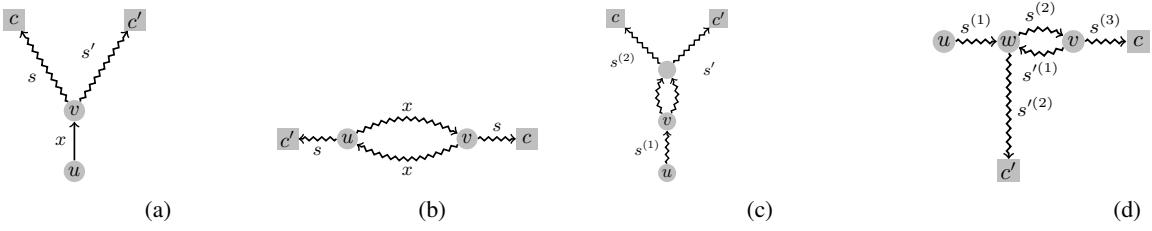
\begin{figure*}[t!]
\centering %\hspace*{0.5cm}
    \begin{subfigure}[b]{0.24\linewidth}
     \resizebox{!}{2.3cm}{
    \begin{tikzpicture}[auto]
        \tikzstyle{vertex}=[circle,fill=black!25,minimum size=10pt,inner sep=0pt]
        \tikzstyle{caster}=[rectangle,fill=black!25,minimum size=10pt,inner sep=0pt]
        \tikzstyle{edge} = [draw,thick,->]
        \tikzstyle{edgewavy} = [draw,thick,->,decorate, decoration={
                                    zigzag,
                                    segment length=4,
                                    amplitude=.9,post=lineto,
                                    post length=2pt
                                }]
        \tikzstyle{weight} = [font=\small]
        
        \foreach \pos/\name/\lab in {{(1,1)/u/u}, {(1,2)/v/v}}
                            \node[vertex] (\name) at \pos{$\lab$};
        \foreach \pos/\name/\lab in {{(0,3.5)/s/c}, {(2,3.5)/t/c'}}
                            \node[caster] (\name) at \pos{$\lab$};
        
        \foreach \source/ \dest /\weight in {v/t/s', v/s/s}
                \path[edgewavy] (\source) -- node[weight] {$\weight$} (\dest);
            \foreach \source/ \dest /\weight in {u/v/x}
                \path[edge] (\source) -- node[weight] {$\weight$} (\dest);
                
    \end{tikzpicture}
    }
    \caption{}\label{sfig: inst1}
    \end{subfigure}
    \hspace{-.8cm} 
    % \hfill
    \begin{subfigure}[b]{0.24\linewidth}
     \resizebox{3.6cm}{!}{
      \begin{tikzpicture}[auto]
        \tikzstyle{vertex}=[circle,fill=black!25,minimum size=10pt,inner sep=0pt]
        \tikzstyle{caster}=[rectangle,fill=black!25,minimum size=10pt,inner sep=0pt]
        \tikzstyle{edge} = [draw,thick,->]
        \tikzstyle{edgewavy} = [draw,thick,->,decorate, decoration={
                                    zigzag,
                                    segment length=4,
                                    amplitude=.9,post=lineto,
                                    post length=2pt
                                }]
        \tikzstyle{weight} = [font=\small]
        
        \foreach \pos/\name/\lab in {{(1,2)/u/u}, {(3,2)/v/v}}
                            \node[vertex] (\name) at \pos{$\lab$};
        \foreach \pos/\name/\lab in {{(0,2)/s/c'}, {(4,2)/t/c}}
                            \node[caster] (\name) at \pos{$\lab$};
        
        \foreach \source/ \dest /\weight in {u/v/x, v/u/x}
                \path[edge] (\source)  edge [edgewavy,bend left=30] node[weight] {$\weight$} (\dest);
        
        \foreach \source/ \dest /\weight in {u/s/s, v/t/s}
                \path[edgewavy] (\source)  -- node[weight] {$\weight$} (\dest);
    \end{tikzpicture} 
    }
    \caption{}\label{sfig: inst2}
    \end{subfigure}
    % \hfill
    \begin{subfigure}[b]{0.24\linewidth}
     \resizebox{!}{2.3cm}{
    \begin{tikzpicture}[auto]
        \tikzstyle{vertex}=[circle,fill=black!25,minimum size=10pt,inner sep=0pt]
        \tikzstyle{caster}=[rectangle,fill=black!25,minimum size=10pt,inner sep=0pt]
        \tikzstyle{edge} = [draw,thick,->]
        \tikzstyle{edgewavy} = [draw,thick,->,decorate, decoration={
                                    zigzag,
                                    segment length=4,
                                    amplitude=.9,post=lineto,
                                    post length=2pt
                                }]
        \tikzstyle{weight} = [font=\small]
        
        \foreach \pos/\name/\lab in {{(1,1)/u/u}, {(1,2)/v/v}, {(1,3)/w/}}
                            \node[vertex] (\name) at \pos{$\lab$};
        \foreach \pos/\name/\lab in {{(0,4)/s/c},{(2,4)/t/c'}}
                            \node[caster] (\name) at \pos{$\lab$};
        
        \foreach \source/ \dest /\weight in {{u/v/s^{(1)}}, {w/s/s^{(2)}}}
                \path[edgewavy] (\source) -- node[weight] {$\weight$} (\dest);
                
                        \foreach \source/ \dest /\weight in {{w/t/s'}}
                \path[edgewavy] (\source) -- node[weight,xshift=.6cm,yshift=-.7cm] {$\weight$} (\dest);
                
                \foreach \source/ \dest /\weight in {v/w/test}
                \path[edge] (\source)  edge [edgewavy, bend left=30] node[weight] {} (\dest);
                
                \foreach \source/ \dest /\weight in {v/w/test}
                \path[edge] (\source)  edge [edgewavy, bend right=30] node[xshift=.8cm] {} (\dest);
    \end{tikzpicture}}\caption{}\label{sfig: inst4}
    \end{subfigure} \begin{subfigure}[b]{0.3\linewidth}
     \resizebox{!}{2.5cm}{
     \begin{tikzpicture}[auto]
        \tikzstyle{vertex}=[circle,fill=black!25,minimum size=10pt,inner sep=0pt]
        \tikzstyle{caster}=[rectangle,fill=black!25,minimum size=10pt,inner sep=0pt]
        \tikzstyle{edge} = [draw,thick,->]
        \tikzstyle{edgewavy} = [draw,thick,->,decorate, decoration={
                                    zigzag,
                                    segment length=4,
                                    amplitude=.9,post=lineto,
                                    post length=2pt
                                }]
        \tikzstyle{weight} = [font=\small]
        
        \foreach \pos/\name/\lab in {{(1,2)/u/u}, {(3,2)/v/v}, {(2,2)/w/w}}
                            \node[vertex] (\name) at \pos{$\lab$};
        \foreach \pos/\name/\lab in {{(2,0)/s/c'}, {(4,2)/t/c}}
                            \node[caster] (\name) at \pos{$\lab$};
        
        \foreach \source/ \dest /\weight in {w/v/s^{(2)}, v/w/s'^{(1)}}
                \path[edge] (\source)  edge [bend left=30,edgewavy] node[weight] {$\weight$} (\dest);
        
        \foreach \source/ \dest /\weight in {w/s/s'^{(2)}, v/t/s^{(3)}, u/w/s^{(1)}}
                \path[edgewavy] (\source)  -- node[weight] {$\weight$} (\dest);
    \end{tikzpicture} 
    }
    \caption{}\label{sfig: inst5}
    \end{subfigure}
    \caption{
    The left two images show constructions as used in the first part of the proof of Proposition~\ref{pr: obsCompIff}. The right two images show the situation in the case distinction in the second part of the proof.
    } \label{fig:prefix-invariant}
\end{figure*}
It holds that $s^{(3)} \in \mathcal{S}_v$ and hence
the choice of $P'$ for $v$ implies that $s' = (s'^{(1)}, s'^{(2)}) \fpre s^{(3)}$.
By using property (i) (multiple times) we obtain: $(s^{(2)}, s'^{(1)}, s'^{(2)}) \fpre (s^{(2)}, s^{(3)})$. Note that, for this to be true these sequences do not have to be feasible sequences within the instance we are considering. 
Moreover, because $s'^{(2)}$ and $s^{(2)}$ are comparable by construction, we have by property (ii): $s'^{(2)} \fpre (s^{(2)}, s'^{(1)}, s'^{(2)})$.
Since all three sequences are pairwise comparable, we can apply the transitivity of $\fpre$ and obtain $s'^{(2)} \fpre (s^{(2)}, s'^{(1)}, s'^{(2)}) \fpre (s^{(2)}, s^{(3)})$. By applying property (i) again (multiple times), we obtain $(s^{(1)}, s'^{(2)}) \fpre (s^{(1)}, s^{(2)}, s^{(3)}) = s$. Since, by the choice of $w$, there exists some path in $\mathcal{P}_u$ with sequence $(s^{(1)}, s'^{(2)})$, this is a contradiction to the choice of the path $P$ for $u$.
\end{proof}
}

\newif\ifprConfluentAppendix 
\prConfluentAppendixtrue

\ifprConfluentAppendix
\else 
\prConfluentProof
\fi

Making use of this characterization we show:
\begin{restatable}[\appSymb]{corollary}{obsBasicConfluent}\label{cor: BFD-MinSUm-confluent}
\bfd and \minsumseq are confluent. \dfd is not confluent.
\end{restatable}

\newcommand{\obsBasicConfluentProof}{
\begin{proof}
Let us first consider BFD and MinSum.
By Proposition~\ref{pr: prefix-invariant}, it is sufficient to check that BFD and \minsumseq satisfy the two properties (i) and (ii). 
Let $\fpre_{BFD}$ and $\fpre_{MinSum}$ denote the corresponding order relations over $\sequences$ for BFD and MinSum, respectively.
For property (i) of Proposition~\ref{pr: prefix-invariant}, consider two comparable sequences $s$ and $s'$ and $x \in \mathbb{N}_{\geq 1}$. If $s \fpre_{BFD} s'$ because $s$ is shorter than $s'$, or $s \fpre_{lex} s'$ in case of equal length, then $(x,s)$ is shorter than $(x,s')$, or $(x,s) \fpre_{lex} (x,s')$, respectively. Thus, $(x,s) \fpre_{BFD} (x,s')$.
Similarly, if $s \fpre_{MinSum} s'$ because $s$ has a lower rank sum than $s'$, or $s \fpre_{lex} s'$ in case of equal rank sum, then $(x,s)$ has a lower rank sum than $(x,s')$, or $(x,s) \fpre_{lex} (x,s')$, respectively. Thus, $(x,s) \fpre_{MinSum} (x,s')$. 
For property (ii) of Proposition~\ref{pr: prefix-invariant}, consider two sequences $s,s' \in \sequences$ such that $s$ and  $(s',s)$ are comparable.
Since $(s',s)$ is longer and also has a higher rank sum than $s$, $s \fpre_{BFD} (s',s)$ and $s \fpre_{MinSum} (s',s)$ trivially holds.

To show that DFD does not satisfy confluence, observe that $(1,2) \lex (2)$. This is clearly a violation of property (ii) in \cref{pr: prefix-invariant}.
\end{proof}
}

\newif\ifobsBasicConfluentAppendix 
\obsBasicConfluentAppendixtrue

\ifobsBasicConfluentAppendix
\else 
\obsBasicConfluentProof
\fi 

\subsection{Advanced Sequence Rules}

\par 

We introduce the new delegation rule \emph{\diffusion} which is inspired by a propagation process similar to those studied in the opinion diffusion literature. We give a motivation for this connection upfront, and then define the rule formally. 

Confronted with an instance of our setting, we argue that there are certain delegation paths which are, in a sense, \emph{best possible}: Let $x$ be the minimum rank of an incoming edge of \emph{any} casting voter. Then, all delegating voters that have a direct edge to a casting voter with rank $x$ should be assigned this path. The rationale behind this statement is that for every voter $v$, every path in $\mathcal{P}_v$ contains at least one edge with rank at least $x$. Hence a one-step path with rank $x$ seems preferable to any other path. 
However, typically, not all voters have such a path. A natural continuation of our argument goes as follows: In a second round, we call casting voters and delegating voters that already got assigned to a casting voter \emph{assigned}.
We treat all assigned voters as casting voters and repeat the process until all delegating voters become assigned. The path assignment is then derived by following the one-step paths. A similar process has been described by \citet{CGN20a} within their unravelling procedures ``basic update'' and ``direct vote priority.''\footnote{Besides the fact that \citet{CGN20a} define their rule for a different setting, they also treat abstaining voters differently as abstentions can be delegated (see also \Cref{fn:unravelling}).}

\smallskip
\noindent\textbf{\diffusion:}
Initialize the set of \emph{assigned} voters: $A \leftarrow C$.

\smallskip

\noindent 
While ($A \neq V \setminus I $), repeat the following steps:
\begin{enumerate}
\item $F \leftarrow \argmin \{r(e) \mid e \in \delta_G^-(A)\}$, where $\delta_G^-(A)$ is the set of edges in $G$ having their head in $A$ and tail in $V\setminus A$.
\item $A \leftarrow A \cup \{v \mid (v,w) \in F\} $
\item $f(v)= ((v,w),f(w)) \text{ for all $(v,w) \in F$}$
\end{enumerate}
    
The assignment in step 3 is well defined as voters' preferences are strict orders and thus there exist no two edges $(v,w), (v,w') \in F$. 
This immediately implies confluence.

\begin{restatable}{proposition}{DiffConfluent}\label{cor: diff-confluent}
\diffusion is confluent. 
\end{restatable}

One may wonder whether this seemingly ``global'' process (in the sense that we need to know the entire graph to determine the minimal rank~$x$) can be explained by an order over $\sequences$ that can be applied to each delegating voter ``locally.'' Stated differently, is diffusion a sequence rule? Surprisingly, we answer this question in the affirmative: We define the order $\difforder$ (which will prove our claim) for sequences without a joint prefix first and then later extend it to any two comparable sequences in a straightforward way. Let $s$ and $s'$ be two comparable sequences with no joint prefix. We define $s \difforder s'$ if one of the following conditions holds:
\begin{enumerate}[labelindent=5pt,label=(\roman*),itemindent=0em,leftmargin=!]
    \item $\max(s) < \max(s')$; 
    \item $\max (s) \!=\! \max (s')$ and $|\argmax(s)|\! < \!|\argmax (s')|$; 
    \item $\max (s)\! = \!\max (s'), |\argmax(s)| \!= \!|\argmax (s')|$,  and  ($\Bar{s} \mathrel{\difforder} \Bar{s}'$ or $\Bar{s}=()$);
\end{enumerate} where $\Bar{s}$ (resp. $\Bar{s}'$) is defined as the prefix of $s$ (resp. $s'$) ending just before the first entry of rank $\max(s)$.

The relation $\difforder$ can now easily be extended to two comparable sequences $(t,s),(t,s')$ having a joint prefix $t\in \sequences$. That is, $(t, s) \mathrel{\difforder} (t,s')$ if and only if $s \difforder s'$.

\begin{restatable}[\appSymb]{theorem}{thmDiffusionSequence} 
The relation $\difforder$ induces \diffusion. \label{diffusionSequence}
\end{restatable}

\newcommand{\thmDiffusionSequenceProof}{
\begin{proof}
We first show that $\difforder$ is a complete order if restricted to any comparable subset of \sequences. To this end, let $(t,s),(t',s)$ be two comparable sequences with joint prefix $t$ (this can be empty). Recall we defined $(t,s) \difforder (t,s') \Leftrightarrow s \fpre s'$. Hence, it remains to show that $s$ and $s'$ are ordered by $\difforder$. 
By (i) and (ii), if $\max(s) < \max(s')$ or $|\argmax(s)| < |\argmax (s')|$, the $s$ and $s'$ are ordered by $\difforder$. To see that $s$ and $s'$ are also ordered by $\difforder$ if $\max(s) = \max(s')$ and $|\argmax(s)| = |\argmax (s')|$, consider sequences $\Bar{s}$ (resp. $\Bar{s}'$) defined as the prefix of $s$ (resp. $s'$) ending just before the first entry of rank $\max(s)$. Either $\Bar{s}$ and $\Bar{s}'$ are trivially ordered by $\difforder$ because of (i) or (ii), which by (iii) implies that $s$ and $s'$ are ordered by $\difforder$, or $\max(\Bar{s}) = \max(\Bar{s}')$ and $|\argmax(\Bar{s})| = |\argmax (\Bar{s}')|$. By repeatedly applying this recursive argument, and because $s$ and $s'$ have no common prefix, case (i) or (ii) are guaranteed to apply at some point such that $s$ and $s'$ are ordered by $\difforder$.

To show that $\difforder$ is transitive, consider three comparable sequences $s, s'$ and $s''$ such that $s \difforder s' \difforder s''$. Assume for contradiction, that $s'' \difforder s$ and that $s,s',s''$ build a counterexample to transitivity that is minimal in $|s|+|s'|+|s''|$. Suppose $s'' \difforder s$ because of (i), i.e.,  $\max(s'') < \max(s)$. Then $\max(s') < \max(s)$ or $\max(s'') < \max(s')$ which is a contradiction to $s \difforder s' \difforder s''$. Now suppose, that $\max(s'') = \max(s)$ but $|\argmax(s'')| < |\argmax (s)|$. Then $\max(s'') = \max(s') = \max(s)$ and either $\argmax(s') < \argmax(s)$ or $\argmax(s'') < \argmax(s')$ which is a contradiction to $s \difforder s' \difforder s''$. Lastly, assume that $\max(s'') = \max(s), |\argmax(s'')| = |\argmax (s)|$  and  $\Bar{s}'' \difforder \Bar{s}$, where $\Bar{s}''$ (resp. $\Bar{s}$) is defined as the prefix of $s''$ (resp. $s$) ending just before the first entry of rank $\max(s'')$. In this case $\max(s'') = \max(s') = \max(s)$ and $|\argmax(s'')| = |\argmax (s')| = |\argmax (s)|$. Let $\Bar{s}'$ be defined as the prefix of $s'$ ending just before the first entry of rank $\max(s')$. Then because $s \difforder s' \difforder s''$, $\Bar{s} \difforder \Bar{s}' \difforder \Bar{s}''$. This is a contradiction to the minimality of the counterexample.

We have shown that $\difforder$ is a linear order on any comparable subset of $\sequences$ and thus $\difforder$ induces a sequence rue. In the following we show that this sequence rule is in fact \diffusion.

Assume for contradiction that there exists an instance $(G,r)$ and a node $v \in D$ with two sequences $s,s' \in \mathcal{S}_v$, such that $s \neq s'$, \diffusion assigns $s'$ to $v$ but $s \difforder s'$. Among all such examples, we choose one which minimizes the joint length of $s$ and $s'$, i.e., $|s| + |s'|$. Using \cref{pr: prefix-invariant} and \cref{lem:helperDifforder} it follows that $\difforder$ is confluent, and because \diffusion is also confluent (Corollary~\ref{cor: diff-confluent}), this implies that $s$ and $s'$ have no joint prefix, since otherwise, there would exist some $v'\in D$, namely at the branching point of $s$ and $s'$, which would induce two sequences with the same order as $s$ and $s'$ but smaller joint length. Let $P$ be the path corresponding to $s$ and $P'$ be the path corresponding to $s'$. 

By the definition of $\difforder$, there exist three possible cases for $s$ being preferred to $s'$.

\smallskip

\noindent\textbf{Case 1.} It holds that $\max(s) < \max(s')$.

Let $e$ be some edge on $P'$ with rank $\max(s')$. Let $i$ be the iteration in which $e \in F$ and denote by $A_i$ the set of assigned nodes at the point in time when $e$ entered $F$ (i.e. in step 1 of the while loop). Denote by $c_1$ the casting voter at the end of $P$. Observe that $v \not\in A_i$ but $c_1 \in A_i$. Hence, there exists an edge on the path $P$ which is in $\delta_G^{-}(A_i)$. Since all edges on $P$ have lower rank than $e$, this is a contradiction to $e$ being in $F$. 

Before continuing with the second case, we state and prove a claim which will be helpful for the remaining two cases. To this end, let $(e_1, \dots, e_k)$ be the edges on $P$ with $r(e_i)=\max(s)$ for all $i \in [k]$, indexed with respect to their appearance on the path $P$. More precisely, $e_k$ is the first edge on $P$ with rank $\max(s)$, i.e., the edge closest to $v$ and $e_1$ is the last edge on $P$ with rank $\max(s)$. Similarly, let $e'_1, \dots, e'_{\ell}$ be the edges on $P'$ with $r(e'_i)=\max(s')$ for all $i \in [\ell]$, indexed with respect to their appearance on the path $P'$. Again, $e'_{\ell}$ is the edge closest to $v$ with rank $\max(s')$ and $e_1$ is the edge furthest from $v$ on $P'$ with rank $\max(s')$. 

\begin{claim} Let $\max(s) = \max(s')$ and $k < \ell$. Let $j \in [k]$ and $i$ be the iteration of the while-loop in which $e'_j \in F$. Then, either $e_j \in F$ or the start node of $e_j$  is already assigned, i.e., is included in $A_i$. 
\end{claim}

\begin{claimproof}
The statement can be shown by induction over $j \in [k]$. For the base case let $j=1$. Let $i$ be the iteration in which $e'_1 \in F$. Let $e_1 = (u_1,v_1)$. If $u_1 \not\in A_i$ and $e_1 \not\in F$, then also $v_1 \not\in A_i$, since otherwise, $e_1 \in \delta_G^-(A_i)$ since $e_1$ has the same rank as $e'_1$. However, from that it follows that some edge on the path from $v_1$ to $c_1$ is included in $A_i$, all of which have rank smaller than $e'_1$, a contradiction to $e'_1$ being in $F$ in iteration $i$.

For the induction step, assume that the statement holds up to some $j \in [k]$. Then, let $i$ be the iteration of the while loop in which $e'_{j+1} \in F$. Let $u_{j}$ and $v_j$ be tail and head of the edge $e_j$, i.e. $e_j=(u_j,v_j)$, and similarly $e_{j+1}=(u_{j+1},v_{j+1})$. By the induction hypothesis we know that $u_j \in A_i$. Now, assume for contradiction that neither $u_{j+1} \in A_i$ nor $e_{j+1} \in F$. Then also $v_{j+1} \not\in A_i$ and there exists at least one edge on the subpath of $P$ from $v_{j+1}$ to $u_i$ which is included in $\delta_G^-(A_i)$. However, by construction, this edge has a smaller rank than the rank of $e'_{j+1}$, a contradiction to $e'_{j+1}$ being in $F$ in iteration $i$. 
\end{claimproof}

\smallskip

\noindent\textbf{Case 2.} $\max(s) = \max(s')$ and $|\argmax(s)| < |\argmax(s')|$. 

With the above shown claim it is easy to arrive at an overall contradiction. Recall that $e_1, \dots, e_k$ are the edges with maximum rank on $P$ and $e'_1, \dots, e'_{\ell}$ are the edges with maximum rank on $P'$. Moreover, $k < \ell$ by the case distinction. Let $i$ be the iteration of the while-loop in which $e'_{k+1} \in F$. By the above claim we know that $u_k$, which is the tail of the edge $e_k$, is included in $A_i$. Moreover, $v \not\in A_i$. Hence, there exists one edge on the subpath of $P$ from $v$ to $u_k$ which is included in $\delta^-(A_i)$. However, as this edge has a rank which is smaller than the rank of $e'_{k+1}$, this is a contradiction to $e'_{k+1}$ being in $F$ in iteration $i$. This concludes the second case.\\

\smallskip

\noindent \textbf{Case 3.} $\max(s) = \max(s')$, $|\argmax(s)| = |\argmax(s')|$ and $\bar{s} \difforder \bar{s}'$, or $\bar{s} = ()$ where $\bar{s}$ (respectively $\bar{s}'$ is defined as the prefix of $s$ (respectively $s'$) ending just before the first entry of $\max(s)$. 

Recall that $e'_k$ is the first edge on the path $P'$ which has a rank of $\max(s')$. Let $i$ be the iteration of the while loop in which $e'_k \in F$ holds. Recall that $e_k=(u_k,v_k)$ is the first edge on the path $P$ with rank $\max(s)$. 
Observe that, by the above claim, it holds that either $u_k \in A_i$ or $e_k \in F$. Now, if $\bar{s}=()$, then $u_k=v$ is assigned at latest after this round. Hence, $v$ is either assigned to $e_k$ or some other edge, but not to the first edge corresponding to the sequence $s'$, a contradiction. 

Thus, we assume in the following that $\bar{s} \difforder \bar{s}'$. We now construct a second reduced instance with graph, $H$, by copying the graph $G$, but defining the set of casting voters to be $A_{i+1}$. To make this a feasible instance, we also delete all outgoing edges of nodes in $A_{i+1}$. Observe that, starting with the iteration $i+1$ for the graph $G$, \diffusion behaves equivalent for the graph $G$ as it does for the graph $H$. More precisely, let $\bar{F}_j$ be the set of the selected edges within the $j^{th}$ iteration of the while-loop within the call of \diffusion for the graph $H$. Similarly, let $F_j$ be the set of the selected edges within $j^{th}$ iteration of the while-loop within the call of \diffusion for the graph $G$. Then, it holds that $\bar{F}_j = F_{j+i+1}$ for $j \in \mathbb{N}$. Because \diffusion selects the path $P'$ for $v$ in $G$, it follows directly that for $v$, \diffusion selects a path with sequence $\bar{s}'$ in $H$.

We show in the following that this contradicts the minimality of the selected counter example as there exists a subpath $\hat{P}$ of $P$, leading to from $v$ to a ``casting'' voter in $H$ with $s(\hat{P}) \difforder \bar{s}'$. To this end, let $w$ be the first node on $P$ which is included in $A_{i+1}$. Then, define $\hat{P}$ to be the prefix of the path $P$ up to $w$. We define $\hat{s} = s(\hat{P})$ and know that $\hat{s}$ is a prefix of $\bar{s}$. \Cref{lem:helperDifforder2} shows that $\bar{s} \difforder \bar{s}'$ implies that $\hat{s} \difforder \bar{s}'$. Since $|\hat{s}| + |\bar{s}'| < |s| + |s'|$ this contradicts the minimality of the counter example, which concludes the third case.
\end{proof}
}

\newif\ifthmDiffusionSequenceAppendix 
\thmDiffusionSequenceAppendixtrue 

\ifthmDiffusionSequenceAppendix
\else 
\thmDiffusionSequenceProof
\fi

The relation $\difforder$ reveals the decision criteria of \diffusion: If two sequences $s$ and $s'$ have a different maximum rank, \diffusion decides in favor of the  sequence with smaller maximum rank (see $(i)$). If the two sequences have equal maximum rank, \diffusion decides in favor of the sequence for which the maximum rank appears less often (see $(ii)$). Thus, \diffusion overcomes \bfd's shortcoming of selecting sequences with large edge ranks. For example, from $\{(1,100), (1,1,2)\}$, \diffusion selects $(1,1,2)$ while \bfd selects $(1,100)$. Moreover, in contrast to \dfd, \diffusion cannot be tricked into selecting a sequence with large edge ranks at the end. For example, \diffusion selects $(2)$ from $\{(1,100),(2)\}$ while \dfd selects $(1,100)$. 
Having said this, the very last tie-breaking rule (see $(iii)$) can be argued to be slightly unnatural as it only compares the parts of $s$ and $s'$ up to the first appearance of the maximum rank. For example, this leads to $(1,5,4,4,4,4) \difforder (2,5)$. 

Inspired by this, we define \emph{\lexrank}, a delegation rule that shares \diffusion's desirable properties but avoids its artificial tie-breaking. 
We define the function $\sigma$, taking as input a rank sequence and sorting the ranks of the sequence in non-increasing order.  
For instance, $\sigma((1,3,4,3)) = (4,3,3,1)$.

\smallskip
\noindent\textbf{\lexrank :} The delegation rule induced by the linear order $\lrorder$ over $\sequences$ defined as follows: For sequences $s,s' \in \sequences$ let $s \lrorder s'$ if (i) $\sigma(s) \lex \sigma(s')$, or (ii) $\sigma(s) = \sigma(s')$ and $s \lex s'$.

\smallskip
In the example in \Cref{fig:running-example}, \lexrank selects $(a,b,c,d,j)$
for voter $a$.
With the help of \cref{pr: prefix-invariant} we can show: 

\begin{restatable}[\appSymb]{corollary}{LexConfluent}
\lexrank is confluent. 
\end{restatable}

\newcommand{\LexConfluentProof}{
\begin{proof}
We start by showing property $(i)$ from \cref{pr: prefix-invariant}. Let $s$ and $s'$ be two comparable sequences from $\sequences$ and $x \in \mathbb{N}_{\geq 1}$. Assume wlog that $s \lrorder s'$. We distinguish three cases: 

\smallskip
\noindent\textbf{Case 1:} $\sigma(s) = \sigma(s')$ and $s \lex s'$. Then, $\sigma((x,s)) = \sigma((x,s'))$ and $(x,s) \lex (x,s')$. Hence $(x,s) \lrorder (x,s')$. 

If case 1 is does not apply, we know that $\sigma(s) \neq \sigma(s')$. In this case, let $i$ be the first index for which $\sigma(s)_i < \sigma(s')_i$ (which is guaranteed to exist due to $\sigma(s) \lex \sigma(s')$).

\smallskip
\noindent\textbf{Case 2:} $x \geq \sigma(s')_i$. Then, $\sigma((x,s))_j = \sigma((x,s'))_j$ for all $j \in \{1, \dots, i\}$ and $\sigma((x,s))_{i+1} < \sigma((x,s'))_{i+1}$. Hence, $\sigma((x,s)) \lex \sigma((x,s'))$ and therefore $(x,s) \lrorder (x,s')$. 

\smallskip
\noindent\textbf{Case 3:} $x < \sigma(s')_i$. Then, $\sigma((x,s))_j = \sigma((x,s'))_j$ for all $j \in \{1, \dots, i-1\}$ and $\sigma((x,s))_{i} < \sigma((x,s'))_{i}$. Hence, $\sigma((x,s)) \lex \sigma((x,s'))$ and therefore $(x,s) \lrorder (x,s')$. 

We now turn to prove property $(ii)$ from \cref{pr: prefix-invariant}. Let $s',s \in \sequences$ such that $s$ and $(s',s)$ are comparable. Clearly $\sigma(s) \neq \sigma(s',s)$. Moreover, since the multiset of elements in $s$ is a subset of the multiset of elements in $(s',s)$, it follows directly that $\sigma(s) \lex \sigma(s',s)$.

Applying \cref{pr: prefix-invariant} concludes the proof that \lexrank is confluent. 
\end{proof}
}

\newif\ifLexConfluentAppendix 
\LexConfluentAppendixtrue

\ifLexConfluentAppendix
\else 
\LexConfluentProof
\fi

\subsection{Branching Rules}
\label{sec:branching}

We can also view the output of a confluent delegation rule as a special \emph{directed forest} (aka \emph{branching}) in the reduced graph $\bar{G}=(D \cup C,\bar{E})$. More precisely, we call $B \subseteq \bar{E}$ a \emph{$C$-branching in $\bar{G}$} if $B$ is acyclic and $|\delta^{+}_B(v)| = 1$ for all $v \in D$. For any confluent rule $f$, the set $\bigcup_{v \in D} f(v)$ is a $C$-branching in $\bar{G}$ and this encodes all selected paths. 

Considering this interpretation, it seems natural to define delegation rules that optimize directly over the set of $C$-branchings, as for example selecting one with minimum sum of ranks. It is important to stress that this approach inherently comes with a different perspective on the problem. Namely, it gives equal importance to each of the edges instead of each of the selected paths. Under the premise that each voter is equally important, this is compatible with the assumption that voters care only about the first edge of their path. 
Though this is different from the approaches presented in the previous sections, this can be a valid view as the voters express their preferences explicitly only over their outgoing edges. Hence, any preferences over paths are ``inherited'' from other voters' preferences.
Below we define \minsumarb, a natural variant of \minsumseq. Such a branching can be found by using a linear program for solving the \emph{min-cost arborescence problem} \citep[e.g.,][]{KoVy06a}.

\smallskip
\noindent \textbf{\minsumarb:} Selects a $C$-branching $B$ in $\bar{G}$ that minimizes $\sum_{e \in B} r(e)$. 
\smallskip

In order for \minsumarb to be resolute, we need to define a tie-breaking rule. 
We later use a \emph{priority order} tie-breaking (formalized in the proof of \cref{prop:bordaGuru}). 
An example shows that \minsumarb is not a sequence rule: Voters $v_1$ and $v_2$ have each other as their first choice and casting voters $w_1$, $w_2$ as their second choice, respectively. Then, $\mathcal{S}_{v_1}=\mathcal{S}_{v_2}=\{(1,2),(2)\}$ and \minsumarb assigns sequence $(1,2)$ to one voter and $(2)$ to the other.

The above approach is reminiscent of the Borda rule in classical social choice theory, as for every branching the method sums up the position of the branching's edges in the corresponding voters' rankings and compares these scores. 

Interestingly, one can also apply an approach in the spirit of Condorcet and consider pairwise majority comparisons between branchings \citep{KKM+21a}. Lifting the delegation preferences of a voter $v \in D$ to preference relation $\succ_v$ over $C$-branchings in a straightforward way (by comparing the ranks of $v$'s outgoing edges in the branchings), define the \textit{majority margin} between two $C$-branchings $B$ and $B'$ as
   \[ \Delta(B,B') = |\{v \in D\mid B \succ_v B'\}| - |\{v \in D\mid B' \succ_v B\}| \text. \]

\noindent A $C$-branching $B$ is called \textit{popular} if $\Delta(B,B')\ge 0$ for all~$B'$. 
It follows from \citet{KKM+21a} that a popular branching in our setting need not exist. They also define the \emph{unpopularity margin}
of $B$ as $\mu(B) = \max_{B'} \Delta(B',B)$. Note that $\mu(B)\ge 0$ for all $B$ and $\mu(B)=0$ iff $B$ is popular. 

In our experiments, we evaluate branchings returned by confluent delegation rules by computing their unpopularity margin via a linear program \citep{KKM+21a}.
Surprisingly, we find that BordaBranching returns a popular branching in most instances.

\section{Axiomatic Analysis} \label{sec:axiomatic}

In this section, we revisit an axiom that was studied in the literature (\emph{guru-participation}) and we formalize a desideratum whose importance was emphasized by practitioners (\emph{copy-robustness}). We also provide axiomatic characterizations of the sequence rules DFD and BFD.

We define the \textit{relative voting weight} $\omega_f(G,r,c)$ of a casting voter $c \in C$ that results from applying a delegation rule $f$ to an instance $(G,r)$ as
\[\omega_f(G,r,c) = \frac{|\{ d \in D \mid f(G,r,d) \text{ ends in } c\}| + 1}{|C| + |D|} \text. \]
We use the \emph{relative} voting weight in order to cope with multiple elections with distinct numbers of non-isolated voters. 

\subsection{Guru-Participation}

\emph{Guru-participation} was introduced by \citet{KoRi20a}, and a similar axiom has been suggested by \citet{BeSw15a}. The axiom demands that a casting voter should not be penalised for being the \guru (also called \emph{guru}) of a delegating voter. More precisely, \citet{KoRi20a} consider a model that also comprises voting on a binary issue, which is decided by majority rule. They say that a casting voter $c$ is ``penalised'' for being the \guru of a delegating voter $v$, if $c$ would prefer the outcome of the election in which, all other things being equal, $v$ abstains. As our model captures the delegation phase only, we need to adapt the axiom to our setting. In the appendix we show that our version implies theirs, and, under a very mild assumption on the delegation rule, the two axioms are equivalent.

\begin{defn}
A delegation rule $f$ satisfies \emph{guru-participa\-tion} if  the following holds for every instance $(G,r)$: 
If $v \in D$ and $f(G,r,v)$ ends in $c$, then 
\[ \omega_f(G,r,u) \leq \omega_f(G',r,u) \text{ for all } u \in C \setminus \{c\}\text,\]  
where $G'=(C'\cup D' \cup I',E')$ is the graph derived from $G=(C \cup D \cup I,E)$ by setting $E' = E \setminus \delta^{+}_G(v)$ and $C'=C$. 
In particular, this implies that $\omega_f(G,r,c) \geq \omega_f(G',r,c)$.
\end{defn}

\citet{KoRi20a} showed that \dfd violates guru-participation while \bfd satisfies it. We generalize the latter statement to all confluent sequence rules. 

\begin{restatable}[\appSymb]{proposition}{confluentGuru}\label{confluentGuru}
Every confluent sequence rule satisfies guru-participation. 
\end{restatable}

\newcommand{\confluentGuruProof}{
\begin{proof}
Within this proof it will be helpful to argue about the absolute voting weight of a casting voter induced by some delegation function $f$ for some instance $(G,r)$. More precisely, this is defined by \[z_f(G,r,u) = |\{d \in D \mid f(G,r,d) \text{ ends in } c \}| + 1 \text.\]

Now, let $v$, $c$, $G$ and $G'$ be defined as in the definition of guru-participation.
Let $\mathcal{P}_w$ and $\mathcal{P}'_w$ be the set of available paths for $w \in D$ in $G$ and $G'$, respectively. Define $\mathcal{S}_w$ and $\mathcal{S}'_w$ analogously. 

Then, note that for all $w \in D, \mathcal P'_{w} \subseteq \mathcal P_{w}$ because $G'$ contains a subset of the edges of $G$. Then, let $w \in D$ and assume that the casting voter of $w$ is not $c$. We define $P$ to be the path for $w$ in $G$, i.e., $P=f(G,r,w)$. By confluence, that means that $v$ is not on the delegation path $P$. Therefore, $P \in \mathcal P'_{w}$. Since $s(P)$ is the maximum element of $\mathcal{S}_{w}$ and $\mathcal{S}'_w \subseteq \mathcal{S}_w$, it is also the maximum element in $\mathcal{S}'_{w}$. Hence, the casting voter of $w$ is unchanged.
Consequently, we have \[z_f(G,r,u) \leq z_f(G',r,u) \text{ for all } u \in C \setminus \{c\} \text.\] 
We now show that this also implies that the same holds for the relative voting weight, i.e.,  

\begin{align*}
    \omega_f(G,r,u) &= \frac{z_f(G,r,u)}{|C|+|D|} \leq \frac{z_f(G',r,u)}{|C|+|D|} \\ 
    & \leq \frac{z_f(G',r,u)}{|C'|+|D'|} = \omega_f(G',r,u)\text, 
\end{align*}
where the second inequality holds because $|C'| + |D'| < |C| + |D|$. This concludes the proof. 
\end{proof}
}

\newif\ifconfluentGuru
\confluentGurutrue

\ifconfluentGuru
\else 
\confluentGuruProof
\fi 

As a consequence, three more rules satisfy the axiom.\footnote{\citet{CGN20a} show that their \emph{unravelling procedures} (two of which reduce to \diffusion up to the fact that abstentions can be delegated) do \textit{not} satisfy guru-participation.
This is an artifact of their treatment of abstaining voters. \label{fn:unravelling}}

\begin{corollary}
\minsumseq, \diffusion, and \lexrank satisfy guru-participation.
\end{corollary}

We show that the same holds for \minsumarb.

\begin{restatable}[\appSymb]{proposition}{bordaGuru}\label{prop:bordaGuru}
There exists a tie-breaking rule for which \minsumarb satisfies guru-participation. 
\end{restatable}

\newcommand{\bordaGuruProof}{
\begin{proof}
We will show the statement for \emph{priority order} tie-breaking, which we define as follows: For a given set of voters $V$, let $\pi: V \rightarrow \{1, \dots, |V|\}$ be an arbitrary bijective function. Then, for two $C$-branchings $B$ and $B'$ in $\Bar{G}$, we say that $B$ is preferred by the priority order over $B'$ (written $B \succ_{\pi} B'$), if and only if there exists $i \in \{1,\dots,|V|\}$ such that
\begin{enumerate}[labelindent=15pt,label=(\roman*),itemindent=3em,leftmargin=!]
    \item $r(\delta_B^{+}(v)) = r(\delta_{B'}^{+}(v)) \text{ for all } v \text{ with } \pi(v) \leq i$, and
    \item $r(\delta_B^{+}(v)) < r(\delta_{B'}^{+}(v)) \text{ for } v \text{ with } \pi(v) = i$.
\end{enumerate}
Note that we slightly abuse notation by applying $r(\cdot)$ to a singleton sets of edges. To avoid technical issues, we define $r(\emptyset)=0$. In particular, $r(\delta_A^{+}(v)) = 0$ for all $v \in I \cup C$. 
Observe that this tie-breaking induces a unique winning $C$-branching for \minsumarb. 

Just as in the proof of \Cref{confluentGuru}, we will show that the statement from the definition of guru-participation is satisfied for the \emph{absolute} voting weight, denoted by $z_f(G,r,u)$ for all $u \in C$. In the proof of \Cref{confluentGuru} we already showed that this suffices to prove guru-participation.  

Now, let $v$, $c$, $G$ and $G'$ be defined as in the definition of guru-participation. We also fix some priority order $\pi$ over $V$. We denote by $\bar{G}$ and $\bar{G}'$ the reduced instance of $G$ and $G'$, respectively.
Let $B$ be the optimal $C$-branching within $(\bar{G},r)$ wrt to \minsumarb with priority order $\pi$ and $B'$ be the optimal $C$-branching within $(\bar{G}',r)$ wrt to \minsumarb with priority order $\pi$. Moreover, let $D_1 \dot\cup D_2 = D$ be a partition of the delegating voters (in the graph $\bar{G}$) such that $D_1$ are those delegating voters with \guru $c$ in branching $B$ and $D_2$ are those delegating voters assigned to some casting voter in $C \setminus \{c\}$. Observe that $|D'| < |D|$ but $C'=C$. In particular, $D_2 \subseteq D'$, i.e., all delegating voters in $D_2$ are still delegating (and not isolated) in $G'$. 

We define $B_i = \{(a,b) \mid a \in D_i, (a,b) \in B\}$ and $B'_i = \{(a,b) \mid a \in D_i, (a,b) \in B'\}$ for $i \in \{1,2\}$. We claim that $B_2 = B'_2$, which would suffice to proof the claim, since all agents in $D_2$ delegate exactly how they did in the original situation, and hence \[z_f(G,r,u) \leq z_f(G',r,u) \text{ for all } u \in C \setminus \{c\}\text.\] 

Assume for contradiction that this is not the case. Then we construct $B'' = B'_1 \cup B_2$ which is, by \cref{lem:merge-branchings} a $C$-branching within the graph $\bar{G}'$. We distinguish two cases. 

\textbf{Case 1:} $\sum_{e \in B''} r(e) > \sum_{e \in B'} r(e)$. Since $B'_1$ is included in both $B'$ and $B''$, this implies that the sum of ranks of $B'_2$ is strictly smaller than the sum of ranks for $B_2$. We will show that this is a contradiction to the optimality of $B$. For this, we construct $\hat{B} = B_1 \cup B'_2$ which, by \cref{lem:merge-branchings} is a $C$-branching within the graph $\bar{G}$. Since $\hat{B}$ and $B$ only differ in $B_2$ and $B'_2$, we get that the sum of ranks of $\hat{B}$ is smaller than the sum of ranks of $B$, a contradiction to the optimality of $B$ in $\bar{G}$. 

\textbf{Case 2:} $\sum_{e \in B''} r(e) = \sum_{e \in B'} r(e)$ and $B' \succ_{\pi} B''$. Let $i \in \{1,\dots, |D'|\}$ be such that the two statements $(i)$ and $(ii)$ from the definition of priority order tie-breaking prove that $B' \succ_{\pi} B''$. It follows that $\pi^{-1}(i) \in D_2$, since agents in $D_1$ get the same outgoing edge in both $B'$ and $B''$. Now, consider once again $\hat{B} = B_1 \cup B'_2$. By the same argument as in case 1, we get that $\hat{B}$ is a $C$-branching in $G$ and $\sum_{e \in \hat{B}} r(e) = \sum_{e \in B} r(e)$. However, it also holds that the two properties $(i)$ and $(ii)$ from the definition of the priority order tie-breaking hold for $i$ and hence prove $\hat{B} \succ_{\pi} B$, which is a contradiction to the optimality of $B$.
This concludes the proof. 
\end{proof}
}

\newif\ifbordaGuruAppendix
\bordaGuruAppendixtrue 

\ifbordaGuruAppendix
\else
\bordaGuruProof
\fi

\subsection{Copy-Robustness}

The issue motivating the next axiom was brought up by \citet{BeSw15a}, who are part of the developing team behind \textit{LiquidFeedback}.
Consider a delegating voter $v$ who is assigned a path of length one, i.e., this voter has a direct connection to its \guru, which we call~$c$. \citet{BeSw15a} argue that there is the threat of a \emph{copy-manipulation}: 
Using communication channels outside the liquid democracy system, these two voters can arrange that voter $v$ acts as a casting voter by copying $c$'s vote. 
If this manipulation leads to a different joint voting weight of $v$ and $c$, then the underlying delegation rule is not \emph{copy-robust}.\footnote{\citet{BeSw15a} consider an even stronger requirement, according to which the final vote count of a binary election needs to remain equal (just as \citeauthor{KoRi20a}, \citeyear{KoRi20a}, they capture the voting phase in their model). We remark that our positive result also holds for this stronger version of the axiom.} 
We formalize this property below.

\begin{defn} \label{def: copy-robustness}
A delegation rule $f$ is \emph{copy-robust} if the following holds for every instance $(G,r)$: If $v \in D$ such that $f(G,r,v)$ is of length one and ends in $c \in C$, then 
\[\omega_f(G,r,c) = \omega_f(G',r,c) + \omega_f(G',r,v), \]
where $G'=(C'\cup D' \cup I',E')$ is derived from $G=(C\cup D \cup I,E)$ by setting $E' = E \setminus \delta^{+}_G(v)$ and $C'=C \cup \{v\}$. 
\end{defn}

We give a consequence of copy-robustness that illustrates how limiting this property is for sequence rules. 

\begin{restatable}[\appSymb]{proposition}{copyRobustChar}
\label{pr: suffix-invariant}
If the sequence rule induced by $\fpre$ is copy-robust, then, for all $x \in \mathbb N$ and for any comparable sequences $s$ and $s'$, it holds that $(s,x) \fpre s' \Leftrightarrow s \fpre s'$. 
\end{restatable} 

\newcommand{\copyRobustCharProof}{
\begin{proof}
 \begin{figure}
      \begin{subfigure}[b]{0.45\linewidth}
      \resizebox{!}{4.6cm}{
     \begin{tikzpicture}
         \tikzstyle{vertex}=[circle,fill=black!25,minimum size=10pt,inner sep=0pt]
         \tikzstyle{caster}=[rectangle,fill=black!25,minimum size=10pt,inner sep=0pt]
         \tikzstyle{edge} = [draw,thick,->]
         \tikzstyle{edgewavy} = [draw,thick,->,decorate, decoration={
                                     zigzag,
                                     segment length=4,
                                     amplitude=.9,post=lineto,
                                     post length=2pt
                                 }]
         \tikzstyle{weight} = [font=\small]
        
         \node[] at (-2, 2){$G$};
        
         \foreach \pos/\name/\lab in {{(0,1)/v/w}, {(0,2)/a/}, {(0,3)/b/}, {(0,4)/c/}, 
                                 {(-1,5)/s/}, {(-2,6.5)/t/v}, 
                                 {(1,5)/d/}}
                             \node[vertex] (\name) at \pos{$\lab$};
         \foreach \pos/\name/\lab in {{(-2,7.5)/u/c}, {(2,6.5)/e/c'}}
                             \node[caster] (\name) at \pos{$\lab$};
        
         \foreach \source/ \dest /\weight in {a/b/, s/t/, d/e/}
                 \path[edgewavy] (\source) -- node[weight, left] {$\weight$} (\dest);
         \foreach \source/ \dest /\weight in {v/a/s_1=s'_1, b/c/s_{k}=s'_{k}, c/s/s_{k+1}, c/d/s'_{k+1}, t/u/x}
                 \path[edge] (\source) -- node[weight, right] {$\weight$} (\dest);
     \end{tikzpicture}
     }
     \caption{}\label{sfig: dfd-characterization a)}
     \end{subfigure}
      \begin{subfigure}[b]{0.45\linewidth}
      \resizebox{!}{4.6cm}{
         \begin{tikzpicture}
         \tikzstyle{vertex}=[circle,fill=black!25,minimum size=10pt,inner sep=0pt]
         \tikzstyle{caster}=[rectangle,fill=black!25,minimum size=10pt,inner sep=0pt]
         \tikzstyle{edge} = [draw,thick,->]
         \tikzstyle{edgewavy} = [draw,thick,->,decorate, decoration={
                                     zigzag,
                                     segment length=4,
                                     amplitude=.9,post=lineto,
                                     post length=2pt
                                 }]
         \tikzstyle{weight} = [font=\small]
        
         \node[] at (-2, 2){$G'$};
        
         \foreach \pos/\name/\lab in {{(0,1)/v/w}, {(0,2)/a/}, {(0,3)/b/}, {(0,4)/c/}, 
                                 {(-1,5)/s/}, {(-2,6.5)/t/}, 
                                 {(1,5)/d/}}
                             \node[vertex] (\name) at \pos{$\lab$};
         \foreach \pos/\name/\lab in {{(-2,7.5)/u/c}, {(-2,6.5)/t/v},{(2,6.5)/e/c'}}
                             \node[caster] (\name) at \pos{$\lab$};
        
         \foreach \source/ \dest /\weight in {a/b/, s/t/, d/e/}
                 \path[edgewavy] (\source) -- node[weight, left] {$\weight$} (\dest);
         \foreach \source/ \dest /\weight in {v/a/s_1=s'_1, b/c/s_{k}=s'_{k}, c/s/s_{k+1}, c/d/s'_{k+1}}
                 \path[edge] (\source) -- node[weight, right] {$\weight$} (\dest);
     \end{tikzpicture}
     }
     \caption{}\label{sfig: dfd-characterization b)}
     \end{subfigure}
     \caption{Construction in the proof of \cref{pr: suffix-invariant}} \label{fig:dfd-characterization}
 \end{figure}
 
Observe that in the definition of copy-robustness, the two instances $(G,r)$ and $(G',r)$ have the same number of non-isolated voters. As a consequence, we can argue about the absolute instead of the relative voting weight in the following. Assume for contradiction that there exists a sequence rule induced by some relation $\fpre$ which is copy-robust but does not satisfy the described property. That is, there exist two comparable sequences violating the property. Among all such pairs, let $(s,x)$ and $s'$ be two sequences minimizing $|(s,x)| + |s'|$. 

We construct a (reduced) graph $(H,w)$ by setting the sequences $(s,x)$ and $s'$ as the two available sequences for voter $v$. Let $(G,r)$ be a problem instance with the property that $(\Bar{G},\Bar{r}) = (H,w)$ (existence guaranteed by our argumentation in \cref{sec:model}). See \Cref{fig:dfd-characterization} for an illustration of the situation where $s=(s_1, \dots, s_{\ell})$, and $s' = (s'_1, \dots, s'_{\ell'})$. Let $(G',r')$ be the instance derived from $(G,r)$ when making $v$, the last delegating voter in the path from $w$ to $c$, a casting voter (see \Cref{sfig: dfd-characterization a)}). 

By the assumption that the pair $\{s,s'\}$ induces a minimal violation against the described property, we know for all delegating voters in $G$ besides $w$ that $c$ is their assigned casting voter within the instance $(G,r)$ if and only $v$ is their assigned casting voter within the instance $(G',r')$.

\textbf{Case 1:} $s' \fpre (s,x)$ but $s \fpre s'$. 
 Hence, $w$'s vote is assigned to casting voter $c'$ within graph $G$ but to $v$ within the graph $G'$. Hence, the number of votes delegated to $c$ and $v$ together increased after acting as if $v$ were a casting voter. This is a contradiction to the fact that the delegation rule induced by $\fpre$ satisfies copy-robustness. 

\textbf{Case 2:} $(s,x) \fpre s'$ but $s' \fpre s$. 
Thus, $w$'s vote is assigned to casting voter $c$ within graph $G$ but to $c'$ within the graph $G'$. Hence, the number of votes delegated to $c$ and $v$ together decreased after acting as if $v$ were a casting voter. This is a contradiction to the fact that the delegation rule induced by $\fpre$ satisfies copy-robustness. 
\end{proof}}

\newif\ifcopyRobustCharAppendix
\copyRobustCharAppendixtrue 

\ifcopyRobustCharAppendix
\else 
\copyRobustCharProof
\fi

\iffalse
\begin{figure}[h!]
    \begin{tikzpicture}[auto]
        \tikzstyle{vertex}=[circle,fill=black!25,minimum size=10pt,inner sep=0pt]
        \tikzstyle{caster}=[rectangle,fill=black!25,minimum size=10pt,inner sep=0pt]
        \tikzstyle{edge} = [draw,thick,->]
        \tikzstyle{edgewavy} = [draw,thick,->,decorate, decoration={
                                    zigzag,
                                    segment length=4,
                                    amplitude=.9,post=lineto,
                                    post length=2pt
                                }]
        \tikzstyle{weight} = [font=\small]
        
        \foreach \pos/\name/\lab in {{(0,1)/u/u}, {(1,2)/v/v}}
                            \node[vertex] (\name) at \pos{$\lab$};
        \foreach \pos/\name/\lab in {{(2,0)/c/c}, {(2,2)/t/c'}}
                            \node[caster] (\name) at \pos{$\lab$};
        
        \foreach \source/ \dest /\weight in {u/c/s_1, u/v/s_2}
                \path[edgewavy] (\source) -- node[weight] {$\weight$} (\dest);
        \foreach \source/ \dest /\weight in {v/t/x}
                \path[edge] (\source) -- node[weight] {$\weight$} (\dest);
    \end{tikzpicture}
\end{figure}
\fi

The restrictive nature of copy-robustness becomes even more apparent in the following impossibility result. 

\begin{restatable}[\appSymb]{proposition}{impConfCopy}
No sequence rule is both confluent and copy-robust. \label{thm:lemma-c-cr}
\end{restatable}

\newcommand{\impConfCopyProof}{
\begin{proof}
Assume there exists a sequence rule induced by $\fpre$, satisfying both confluence and copy-robustness. 
Consider the instances in Figure~\ref{fig:lemma-c-cr} with casting voters $t$ and $s$.
In both instances, $t$ is the assigned casting voter for $v$ and $s$ is the assigned casting voter for $u$. Assume for contradiction that $s$ is the assigned casting voter for $v$. Because $\mathcal{S}_u = \mathcal{S}_v$ in any of the two networks, then $t$ is the assigned casting voter for $u$. This contradicts confluence. 

In Figure~\ref{sfig: lemma-c-cr a)}, if $u$ becomes a casting voter, $t$ needs remain to be the assigned casting voter for $v$ by copy-robustness. This implies that $(2) \fpre (1)$. The same reasoning on the instance in Figure~\ref{sfig: lemma-c-cr b)} gives $(1) \fpre (2)$, a contradiction.
\begin{figure}
     \begin{subfigure}[b]{0.45\linewidth}
     \begin{center}
     \resizebox{3.5cm}{!}{
    \begin{tikzpicture}[auto]
        \tikzstyle{vertex}=[circle,fill=black!25,minimum size=10pt,inner sep=0pt]
        \tikzstyle{caster}=[rectangle,fill=black!25,minimum size=10pt,inner sep=0pt]
        \tikzstyle{edge} = [draw,thick,->]
        \tikzstyle{edge_bend} = [draw,thick,->]
        \tikzstyle{weight} = [font=\small]
         
        \foreach \pos/\name/\lab in {{(1,1)/v/v}, {(2,1)/u/u}}
                            \node[vertex] (\name) at \pos{$\lab$};
        \foreach \pos/\name/\lab in {{(0,1)/t/t}, {(3,1)/s/s}}
                            \node[caster] (\name) at \pos{$\lab$};
        
        \foreach \source/ \dest /\weight in {v/t/2, u/s/2}
                \path[edge] (\source) edge node[weight] {$\weight$} (\dest);
         
         \foreach \source/ \dest /\weight in {v/u/1, u/v/1}
                \path[edge] (\source) edge [bend left=30] node[weight] {$\weight$} (\dest);
    \end{tikzpicture}
    }
     \end{center}
    \caption{}\label{sfig: lemma-c-cr a)}
    \end{subfigure}
     \begin{subfigure}[b]{0.45\linewidth}
     \begin{center}
     \resizebox{3.5cm}{!}{
        \begin{tikzpicture}[auto]
        \tikzstyle{vertex}=[circle,fill=black!25,minimum size=10pt,inner sep=0pt]
        \tikzstyle{caster}=[rectangle,fill=black!25,minimum size=10pt,inner sep=0pt]
        \tikzstyle{edge} = [draw,thick,->]
        \tikzstyle{edge_bend} = [draw,thick,->]
        \tikzstyle{weight} = [font=\small]
         
        \foreach \pos/\name/\lab in {{(1,1)/v/v}, {(2,1)/u/u}}
                            \node[vertex] (\name) at \pos{$\lab$};
        \foreach \pos/\name/\lab in {{(0,1)/t/t}, {(3,1)/s/s}}
                            \node[caster] (\name) at \pos{$\lab$};
        
        \foreach \source/ \dest /\weight in {v/t/1, u/s/1}
                \path[edge] (\source) edge node[weight] {$\weight$} (\dest);
         
         \foreach \source/ \dest /\weight in {v/u/2, u/v/2}
                \path[edge] (\source) edge [bend left=30] node[weight] {$\weight$} (\dest);
    \end{tikzpicture}
    }
     \end{center}
    \caption{}\label{sfig: lemma-c-cr b)}
    \end{subfigure}
     \caption{Construction in the proof of Lemma \ref{thm:lemma-c-cr}.} \label{fig:lemma-c-cr}
\end{figure}
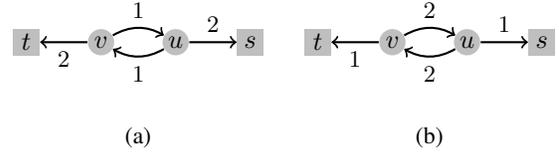
\end{proof}
}

\newif\ifimpConfCopyAppendix
\impConfCopyAppendixtrue 

\ifimpConfCopyAppendix
\else
\impConfCopyProof
\fi 

Hence, none of \bfd, \minsumseq, \diffusion, and \lexrank is copy-robust. On the other hand, \citet{BeSw15a} have implicitly shown that the non-confluent sequence rule \dfd satisfies the axiom (we give a proof in  \cref{thm:dfd-characterization}). 
Our next result shows that the incompatibility between confluence and copy-robustness is indeed restricted to sequences rules: We prove that \minsumarb is copy-robust. We use the same tie-breaking as in \Cref{prop:bordaGuru}.

\begin{restatable}[\appSymb]{proposition}{arbCopy}
There exists a tie-breaking rule for which \minsumarb satisfies copy-robustness.
\end{restatable}

\newcommand{\arbCopyProof}{\begin{proof}

In the following we show that \minsumarb satisfies copy-robustness when using the \emph{priority order} tie-breaking, as defined in the proof of \Cref{prop:bordaGuru}. Let $(\Bar{G},r)$ be the reduced graph of some input instance for which there exists some delegating voter $v \in D$, such that, within the $C$-branching $B$ selected by \minsumarb within the graph $\Bar{G}$, the head of the edge $\delta_{B}^{+}(v)$ is a casting voter, which we will call $c$. Now, let $(G',r)$ be the instance derived from $G$ when $v$ becomes a casting voter (as defined within \cref{def: copy-robustness}). Let $B'$ be the $C$-branching selected by \minsumarb within the graph $\Bar{G}'$ (the reduced graph of $G'$). We claim that $B' = B \setminus \{(v,c)\}$. Assume not. Then, we construct an $C$-branching $\hat{B}$ within the graph $\bar{G}$ by setting $\hat{B} = B' \cup \{(v,c)\}$. Observe that this is possible since $v$ is a sink in $\bar{G}'$. We will show that $\hat{B}$ constitutes a contradiction to the selection of $B$ in $\bar{G}$. To this end, we consider the following two cases. 

\textbf{Case 1:} $\sum_{e \in B'} r(e) < \sum_{e \in B \setminus \{(v,c)\}} r(e)$. 
Then,
\begin{align*}
\sum_{e \in \hat{B}} r(e) & = \sum_{e \in B'} r(e) + r((v,c)) \\
 & <\sum_{e \in B \setminus \{(v,c)\}} r(e) + r((v,c)) = \sum_{e \in B} r(e),
\end{align*}
a contradiction to the minimality of $A$ in $\bar{G}$. 

\textbf{Case 2:} $B' \succ_{\sigma} B \setminus \{(v,c)\}$ and 
$\sum_{e \in B'} r(e) = \sum_{e \in B \setminus \{(v,c)\}} r(e)$. 
Let $i \in \{1, \dots |V|\}$ be such that the two statements $(i)$ and $(ii)$ from the definition of lexicographical tie-breaking prove that $B' \succ_{\sigma} B \setminus \{(v,c)\}$. From that it follows that $\hat{B} \succ_{\sigma} B$ since, no matter whether $\sigma(v)<i$ or $\sigma(v)>i$, the statements $(i)$ and $(ii)$ hold true for $i$, $\hat{B}$ and $B$. Moreover, due to the same argument as in Case 1, we get that $\sum_{e \in \hat{B}} r(e) = \sum_{e \in B} r(e)$, which is a contradiction to the fact that $B$ is selected within the graph $\bar{G}$.
\end{proof}
}

\newif\ifarbCopyAppendix
\arbCopyAppendixtrue %set to false if wanted

\ifarbCopyAppendix
\else
\arbCopyProof
\fi

\subsection{Characterizations} \label{subsec: impossibilities}

Finally, we give axiomatic characterizations of DFD and BFD. To this end, we define two additional properties. 

\begin{defn}
A sequence rule $f$ is \emph{weakly lexicographic} if for two comparable sequences $s$ and $s'$ with $|s|=|s'|$ that only differ in their last rank value, it holds that $s \fpre s' \Leftrightarrow s \lex s'.$ The rule $f$ is \emph{strongly lexicographic} if the same holds for any two comparable sequences of equal length. 
\end{defn}

Every reasonable sequence rule should be weakly lexicographic, as this can also be seen as avoiding the selection of ``Pareto dominated'' sequences. The stronger version, however, is technical in nature and has no normative appeal.

\begin{restatable}[\appSymb]{theorem}{DFDChar}
DFD is the only sequence rule that is weakly lexicographic and copy-robust.
\label{thm:dfd-characterization}
\end{restatable}

\newcommand{\DFDCharProof}{
\begin{proof}
It follows directly from the definition of DFD that it is a sequence rule and weakly lexicographic. We now show that DFD also satisfies copy-robustness. Let $(G,r)$ be an instance in which some delegating voter $v$ is assigned a one-edge path by DFD. We call $c$ the casting voter that $v$ is assigned to. Let $(G',r)$ be the instance that is obtained from $(G,r)$ when $v$ becomes a casting voter (as defined in the definition of copy-robustness). Denote by $\mathcal{S}_w$ and $\mathcal{S}'_w$ be the sequences corresponding to paths from some delegating voter $w$ to casting voters in $(G,r)$ and $(G',r)$, respectively. We claim $DFD(G,r,w)$ ends in $c$ if and only if $DFD(G',r,w)$ ends in $c$ or $v$. This suffices to prove that $DFD$ is copy-robust. In the following, we distinguish two cases: 

\textbf{Case 1.} $DFD(G,r,w)$ does not contain $v$. The set $\mathcal{S}'_w$ consists of the same sequences as $\mathcal{S}_w$ with the difference that any sequence induced by a path leading through $v$ is cut off at this point. Let $s$ be such that sequence and $s'$ be the corresponding sequence after the cut-off at $v$. We know that $s(DFD(G,r,w)) \lex s$. Since $s'$ is not a prefix of $s(DFD(G,r,w))$ we also know that $s(DFD(G,r,w)) \lex s'$. Hence $s(DFD(G,r,w))$ is still lexicographically optimal in $\mathcal{S}'_w$ and $w$'s vote is assigned to the same casting voter in $G$ and $G'$. 

\textbf{Case 2.} $DFD(G,r,w)$ contains $v$. 
Since the edge $(v,c)$ is lexicographically smallest among all edges not leading to isolated voters, it holds that $(v,c) \in DFD(G,r,w)$, i.e., the voting weight of voter $w$ is assigned to casting voter $c$ in $G$. 
Moreover, the sequence of $DFD(G,r,v)$ without its last edge is lexicographically smaller than all sequences in $\mathcal{S}_v$. 
The set $\mathcal{S}'_w$ consists of the same sequences as $\mathcal{S}_w$ with the difference that any sequence induced by a path leading through $v$ is cut off at this point. The sequence of $DFD(G,r,w)$ without the last edge is also lexicographically minimal among $\mathcal{S}'_w$ and the voting weight of $w$ is assigned to $v$ in $G'$.

Now, assume for contradiction that there exists a sequence rule $f$ with order relation $\fpre_f$ that is different from the order relation $\fpre_{DFD}$ of DFD, is weakly lexicographic and copy-robust. Then there exist at least two comparable sequences $s$ and $s'$ such that $s \fpre_{f} s'$ but $s' \fpre_{DFD} s$.
Let $k$ be the length of the joint prefix of $s$ and $s'$.
Because both $f$ and DFD are copy-robust, we can repeatedly apply \cref{pr: suffix-invariant}, i.e.,
\begin{align*}
    & s \fpre_{f} s' \\ 
    \Leftrightarrow & (s_1, \dots, s_{k+1}) \fpre_{f} s' \\ 
    \Leftrightarrow & (s_1, \dots, s_{k+1}) \fpre_{f} (s'_1, \dots, s'_{k+1}) \text{, and }\\
  & s' \fpre_{DFD} s \\ 
   \Leftrightarrow &(s'_1, \dots, s'_{k+1}) \fpre_{DFD} s \\
   \Leftrightarrow &(s'_1, \dots, s'_{k+1}) \fpre_{DFD} (s_1, \dots, s_{k+1}) \text.   
 \end{align*}
Since $f$ and \dfd have to decide equally for the two sequences $(s_1, \dots, s_{k+1})$ and $(s'_1, \dots, s'_{k+1})$, this is a contradiction. 
\end{proof}
}

\newif\ifDFDCharAppendix
\DFDCharAppendixtrue

\ifDFDCharAppendix
\else
\DFDCharProof
\fi

\begin{restatable}[\appSymb]{theorem}{BFDChar}\label{thm:bfd-characterization}
BFD is the only sequence rule that is confluent and strongly lexicographic.
\end{restatable}

\newcommand{\BFDCharProof}{
\begin{proof}
By \Cref{cor: BFD-MinSUm-confluent} BFD is confluent and by definition BFD is strongly lexicographic. In the following we show that BFD is the only such sequence rule.

Let $f$ be a sequence rule that is confluent and strongly lexicographic with order relation $\fpre_f$ over sequences. We will show in the following that $\fpre_f$ equals the sequence relation $\fpre_{\bfd}$ induced by BFD. 
For contradiction, assume that $\fpre_f$ and $\fpre_{\bfd}$ are different. 
 Then there exist at least two comparable sequences $s$ and $s'$ such that $s \fpre_{f} s'$ but $s' \fpre_{\bfd} s$. Among all such pairs $s$ and $s'$, choose one such that $|s|+|s'|$ is minimal.
 Because $f$ and BFD are strongly lexicographic, $s$ and $s'$ cannot have the same length, since otherwise both would order the sequences in the same way.
 Since BFD prefers short sequences, we know that $|s| > |s'|$. 
 Denote sequence $s = (s_1, \dots, s_\ell)$ and $s' = (s'_1, \dots, s'_{\ell'})$ such that $\ell > \ell'$.
Now, heading towards a contradiction, 
consider the reduced instance depicted in Figure~\ref{fig:bfd-characterization} and observe that $\mathcal{S}_{v_1} = \{(1, \dots, 1, 2, s), (1, \dots, 1, 2,s'),(2,s),(2,s')\}$. From our assumption we know that $s(f(v_2))=s$. Since $f$ is confluent it follows that $s(f(v_1)) \in \{(1, \dots, 1, 2,s),(2,s)\}$. 
Furthermore, because of confluence and by Proposition~\ref{pr: prefix-invariant} (ii), $s(f(v_1)) = (2,s)$, because $(2,s)$ is a suffix of $(1, \dots, 1, 2,s)$ and the sequences are comparable.
However, in particular this implies $(2,s) \fpre_f (1, \dots, 1, 2,s')$. Because $(1,...,1,2,s)$ and $(1,...,1,2,s')$ are of equal length, this is a contradiction to $f$ being strongly lexicographic.
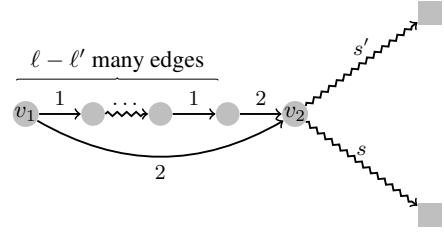
\begin{figure}
\begin{center}
     \resizebox{!}{3cm}{
    \begin{tikzpicture}
        \tikzstyle{vertex}=[circle,fill=black!25,minimum size=10pt,inner sep=0pt]
        \tikzstyle{caster}=[rectangle,fill=black!25,minimum size=10pt,inner sep=0pt]
        \tikzstyle{edge} = [draw,thick,->]
        \tikzstyle{edgewavy} = [draw,thick,->,decorate, decoration={
                                    zigzag,
                                    segment length=4,
                                    amplitude=.9,post=lineto,
                                    post length=2pt
                                }]
        \tikzstyle{weight} = [font=\small]
        
        \foreach \pos/\name/\lab in {{(1,0)/v/v_1}, {(2,0)/a/}, {(3,0)/b/}, {(4,0)/c/}, {(5,0)/d/v_2}}
                            \node[vertex] (\name) at \pos{$\lab$};
        \foreach \pos/\name/\lab in {{(7,-1.5)/u/}, {(7,1.5)/w/}}
                            \node[caster] (\name) at \pos{$\lab$};
        
        \foreach \source/ \dest /\weight in {a/b/\dots, d/u/s, d/w/s'}
                \path[edgewavy] (\source) -- node[weight, above] {$\weight$} (\dest);
        \foreach \source/ \dest /\weight in {v/a/1, b/c/1, c/d/2}
                \path[edge] (\source) -- node[weight, above] {$\weight$} (\dest);
        \foreach \source/ \dest /\weight in {v/d/2}
\path[edge] (\source)  edge [bend right =30] node[weight, below] {$\weight$} (\dest);
        \draw [decorate,decoration={brace,amplitude=1pt},xshift=-4pt,yshift=0pt]
(1,0.5) -- (4,0.5) node [black,midway, above]{$\ell-\ell'$ many edges};
    \end{tikzpicture}
    }
\end{center}
    \caption{Construction in the proof of Theorem \ref{thm:bfd-characterization}.} \label{fig:bfd-characterization}
\end{figure}
\end{proof}

}

\newif\ifBFDCharAppendix
\BFDCharAppendixtrue

\ifBFDCharAppendix
\else
\BFDCharProof
\fi

The appendix contains a characterization of \diffusion and an overview of our axiomatic results (\cref{tab:axioms}).

 \begin{figure*}[ht]
\begin{subfigure}{0.49\textwidth}
    \centering
    \input{fig_friends2}
    \caption{Synthetic data, averaged over $1000$ instances ($n=1000$, $\Delta = 4$, $p_c = 0.2$,  $\alpha = 2$)}
    \label{fig:table_friendship_synthetic}
\end{subfigure}\hfill
\begin{subfigure}{0.49\textwidth}
    \centering
    \input{fig_Facebook-reduced}
    \caption{Real-world data using Facebook network ($n = 63,731$ and $m=817,035$), averaged over $10$ instances ($p_c = 0.2$, $\alpha = 1$) }
    \label{fig:table_popularity_synthetic}
\end{subfigure}
\caption{Evaluation of delegation rules with respect to several quantities for real and synthetic unweighted undirected networks.} \label{fig:table_results_exp}
\end{figure*}

\section{Experimental Evaluation} \label{sec:experiments}

To complement our axiomatic results, we conducted an extensive experimental study in order to  compare delegation rules with respect to quantitative criteria.
To the best of our knowledge, there are no real-world data sets on liquid democracy with ranked delegations, nor do there exist established data generation methods for this setting. Thus, we developed three distinct \textit{methods for generating instances} of our setting. 
In all cases, we transform a (directed or undirected) \emph{base graph} $H = (V_H, E_H)$ to an instance $(G, r)$ of our setting,  with $G = (C \cup D \cup I, E)$. While it is always the case that $ V_H = C \cup D \cup I$, the set of casting voters $C \subseteq V_H$, the edges of the graph $E \subseteq E_H$, and the rank function $r$ are selected in a random process specified by the generation method. Hence, one base graph $H$ paired with one of the three methods leads to a wide range of instances. For each method, we use multiple base graphs coming from synthetic and real-world networks (such as partial networks of \emph{Facebook} and \emph{Twitter}). We let $n$ (respectively, $m$) denote the number of nodes (respectively, edges) of the graph $G$.
             
Our experimental setup is described in more detail in the appendix, where we also present the results for different generation methods, datasets, and parameters. Our code can be found at \url{https://github.com/TheoDlmz/rankeddelegation}.

\paragraph{Data Generation}

Our three generation methods apply to different types of base graphs: unweighted undirected networks, unweighted directed networks, and weighted directed networks. We describe the first method below, and refer to the appendix for descriptions of the other methods.\footnote{
In our method for unweighted directed networks, which is inspired by \citet{GKMP18b}, voters are more likely to delegate to nodes with a high indegree. This is reminiscent of the \textit{preferential attachment} model \citep{BaAl99a}.
For weighted directed networks, we interpret weights as ``trust'' and rank-order potential delegation edges by decreasing weight.} 

When the base graph $H = (V_H, E_H)$ is undirected and unweighted, we assume edges to represent ``friendship'' links between the nodes. We quantify the ``strength'' of a friendship between two adjacent nodes $v$ and $w$ by the number of common friends, i.e., $\lambda(v, w) = \left |  \delta_H (v) \cap \delta_H (w) \right|$, where $\delta_H(\cdot)$ is the set of neighbors of a node in $H$. 
We first choose $C \subseteq V_H$ by selecting each voter independently with probability $p_c \in [0,1]$. 
Then, for every $v \in V\setminus C$, we insert edges in $\delta_H(v)$ one by one as out-going edges for $v$ in $G$ where the probability of selecting the edge $\{v,w\} \in E_H$ is proportional to $(1 + \lambda(v,w))^{\alpha}$, and $\alpha>0$ is a constant.
The ranking $r$ among edges is defined by the order of selection. 

The real-world network we used for this method is a subgraph of the Facebook network \citep{viswanath-2009-activity}. 
For synthetic networks, the input network $H$ is generated by the standard Erd\H{o}s--R\'{e}nyi model $G(n,p)$, where we chose the edge probability $p \in [0,1]$ such that the average number of edges per voter is equal to $\Delta$.

\paragraph{Evaluation Metrics}

We implemented all considered delegation rules and evaluated the output of the delegation rules on various metrics: The maximum rank found on any delegation path (\textit{MaxRank}), the maximum and average length of delegations paths (\textit{MaxLen} and \textit{AvgLen}), and the maximum sum of ranks of a delegation path (\textit{MaxSum}).
We also computed \textit{MaxWeight}, the maximum value of the \textit{relative voting weight} 
$\omega_f(G,r,c)$
of a casting voter $c$, as a measure of the balancedness of the distribution of voting weight.  
(This value plays a crucial role in the analysis of \citealp{GKMP18b}.) 
Finally, for confluent rules, we computed the average rank of outgoing edges (\textit{AvgRank}) and the unpopularity margin of the selected branchings, divided by the number of non-isolated voters (\textit{Unpop}).
For all of these metrics, small values are desirable and correspond to light colors in our tables.

\paragraph{Experimental Results}
 
Figure \ref{fig:table_results_exp} presents the results for unweighted undirected base networks.
These results are representative also for other generation methods and data sets.  

The results indicate that the rules can be roughly aligned along a spectrum: 
On the one extreme, \bfd leads to a good weight distribution (i.e., low average values of MaxWeight), but high (maximum and average) ranks. 
On the other extreme, \dfd and \minsumarb perform better on the rank metrics but often lead to unbalanced weight distributions and long delegation paths. 
Other sequence rules fall between these extremes, with  \minsumseq closer to \bfd, and  \lexrank and \diffusion closer to \dfd and \minsumarb.
\minsumseq in particular seems to strike an attractive balance among most of the considered evaluation criteria. 

Other noteworthy observations include that \lexrank outperforms Diffusion on every metric (recall that we introduced the former as a simplification of the latter) and that \minsumarb very often produces popular branchings (for example, this is the case in $93\%$ of instances for the synthetic data set in \Cref{fig:table_results_exp}(a)).
On average, \minsumarb mostly outperforms \dfd, and it also has the advantage of being a confluent rule that moreover satisfies guru-participation and copy-robustness.

\section{Discussion}

Building upon a graph-theoretic model, we explored the rich space of delegation rules
and introduced novel rules that axiomatically and empirically outperform existing ones. 

Our experiments revealed a gap between rules such as \bfd and \minsumseq on the one side, and rules such as \lexrank on the other. This gap can be filled by defining \textit{weighted} sequence rules: 
For an increasing function $w: \mathbb{N}_{\geq 1} \rightarrow \mathbb{R}^+$, such a rule orders sequences by 
$\sum_{i} w(s_i)$.
Our axiomatic results for \minsumseq can be generalized to this class.

In future work,
it might be desirable to take preferences over delegation \textit{paths} (rather than only edges) into account. 
These could then be lifted to preferences over branchings, and 
approaches like those described in \Cref{sec:branching} could still be used to define confluent delegation rules. 

While some branching rules are inherently non-neutral, \emph{randomized delegation rules} have the potential to avoid this issue by selecting probability distributions over outgoing edges \citep{Bril18a}. 
Going beyond deterministic delegations would also raise the necessity for new axioms. 

Our experiments indicate a trade-off between minimizing the quantities \emph{unpopularity} and \emph{average rank} on the one hand, and \emph{maximum length} and \emph{maximum weight} on the other hand. It would be interesting to formalize this trade-off by finding worst-case bounds and guarantees.

We also suggest a comparison between outcomes of liquid democracy with and without ranked delegations.

\section*{Acknowledgments}

We would like to thank Jan Behrens, Axel Kistner, Andreas Nitsche, and Bj\"{o}rn Swierczek from the Association for Interactive Democracy for helpful discussions. 
This work was supported by the Deutsche Forschungsgemeinschaft under grant BR~4744/2-1, by the Austrian Science Fund (FWF), project P31890, and by the Norwegian  Research Council, project nr. 302203.

\clearpage
\appendix

\newif\ifAppendix
\Appendixtrue

\ifAppendix
\section{Missing Proofs of Section \ref{sec:rules}}

\obsCompIff*

\ifobsCompIffAppendix
\obsCompIffProof
\fi

\prConfluent* 

\ifprConfluentAppendix 
\prConfluentProof
\fi

\obsBasicConfluent*

\ifobsBasicConfluentAppendix
\obsBasicConfluentProof 
\fi

Before we turn to proof \cref{diffusionSequence}, we show two helpful lemma for $\difforder$.  

\begin{lemma}
The relation $\difforder$ satisfies the two properties in \cref{pr: prefix-invariant}. \label{lem:helperDifforder}
\end{lemma}

\begin{proof}
We first show property (i). Let $s,s'$ be two comparable sequences and $x \in \mathbb{N}_{\geq}$. Let $s=(s^{(1)},s^{(2)})$ and $s'=(s^{(1)},s'^{(2)})$ such that $s^{(1)}$ is the joint prefix of $s$ and $s'$ (possibly empty). By definition of \difforder, we have that $(x,s^{(1)},s^{(2)}) \difforder (x,s^{(1)},s'^{(2)}) \Leftrightarrow (s^{(2)}) \difforder (s'^{(2)}) \Leftrightarrow (s^{(1)},s^{(2)}) \difforder (s^{(1)},s'^{(2)})$, this suffices to prove the property. 

For property (ii) assume for contradiction that there exist $s,s' \in \sequences$ such that $s$ and $(s',s)$ are comparable and $(s',s) \difforder s$. Among all such violations chose one that minimizes $|s|$ + $|s'|$. It follows directly that $\max(s') < \max(s)$ (otherwise, we would already have a contradiction). Let $\bar{s}$ be the prefix of the $s$ up to the first appearance of $\max(s)$. If $\bar{s} = ()$, then $s \difforder (s',s)$, a contradiction. Otherwise, $(s',s) \difforder s$ implies that $(s',\bar{s}) \difforder s$, a contradiction to the choice of $s'$ and $s$.  
\end{proof}

\begin{lemma}
Let $s,s',t \in \sequences$ be such that $s'$ is a prefix of $s$ and $s \difforder t$. Then, $s' \difforder t$. \label{lem:helperDifforder2}
\end{lemma}

\begin{proof}
We distinguish three cases. 

\textbf{Case 1:} $\max(s') < \max(s)$. Clearly, $\max(s') < \max(s) \leq \max(t)$ and $s' \difforder t$.%, a contradiction. 

\textbf{Case 2:} $\max(s') = \max(s)$ and $|\argmax(s')| < |\argmax(s)|$. Then, either $\max(s') < \max(t)$, in which case $s' \difforder t$, or $\max(s') = \max(t)$ and $|\argmax(s')| < |\argmax(s)| \leq |\argmax(t)|$, in which case also $s'
 \difforder t$. 

 \textbf{Case 3:} $\max(s') = \max(s)$ and $|\argmax(s')| = |\argmax(s)|$. If $s \difforder t$ because of one of the first two cases in the definition of \difforder, we are done. Otherwise, observe that $\bar{s}' = \bar{s}$ and hence $\bar{s'} \difforder \bar{t}$ and $s' \difforder t$. 
\end{proof}

\thmDiffusionSequence* 

\ifthmDiffusionSequenceAppendix
\thmDiffusionSequenceProof
\fi

\LexConfluent*
\ifLexConfluentAppendix
\LexConfluentProof
\fi

%%%%%%%%%%%%%%%%%%%%%%

\begin{table*}[t]
    \centering
    \begin{tabular}{lcccccc}
    \toprule
         & \phantom{x} \bfd \phantom{x} & \phantom{x} \minsumseq \phantom{x} & \lexrank & \diffusion & \minsumarb & \phantom{x} \dfd \phantom{x}  \\
    \midrule
    Confluent & \yes & \yes & \yes & \yes & \yes & \no \\
    Sequence rule & \yes & \yes & \yes & \yes & \no & \yes\\
    Guru-participation & \yes & \yes & \yes & \yes & \yes & \no\\
    Copy-robustness & \no & \no & \no & \no & \yes & \yes \\ 
    Weakly lexicographic & \yes  & \yes & \yes & \yes & n/a & \yes\\
    Strongly lexicographic & \yes  & \no & \no & \no & n/a & \yes\\ 
    Rank-aware & \no & \no & \yes & \yes & n/a & \no \\
    \bottomrule
    \end{tabular}
    \caption{Summary of the properties of the delegation rules. The latter three properties are only defined for sequence rules.}  
    \label{tab:axioms}
\end{table*}

\section{Missing Proofs of Section \ref{sec:axiomatic}}

We start this section by formalizing the ``two-stage'' (delegation and voting) setting proposed by \citet{KoRi20a}, using our own notation. 
Every casting voter $c \in C$ has an opinion on a binary issue, which we denote by $\tau_c \in \{0,1\}$. Using the majority rule, option $0$ is implemented if \[\sum_{c \in C} \omega_{f}(G,r,c) \cdot \tau_c < 0.5, \] option $1$ is implemented if \[\sum_{c \in C} \omega_{f}(G,r,c) \cdot \tau_c > 0.5\] and a tie (denoted by $0.5$) is obtained if \[\sum_{c \in C} \omega_{f}(G,r,c) \cdot \tau_c = 0.5 \text.\]
We assume that a casting voter $c$ with $\tau_c = 0$ has the following preferences over the outcome $0 \succ_c 0.5 \succ_c 1$. Similarly, a casting voter with $\tau_c = 1$ has preferences $1 \succ_c 0.5 \succ_c 0$.\footnote{\citet{KoRi20a} did not specify the preferences in case of a tie. Having said this, our argument would also work for any other tie-breaking rule.} 

We are now ready to define the axiom as stated by \citet{KoRi20a} and mark it by * to distinguish it from the property defined in \Cref{sec:axiomatic}.

\begin{defn}
A delegation rule $f$ satisfies \emph{guru-participa\-tion*} if  the following holds for every instance $(G,r)$: 
If $v \in D$ and $f(G,r,v)$ ends in casting voter $c$, then $o = o'$ or $o \succ_c o'$ where $o$ is the election outcome of instance $(G,r)$ and $o'$ is the election outcome of $(G',r)$. The graph $G'=(C'\cup D' \cup I',E')$ is derived from $G=(C \cup D \cup I,E)$ by setting $E' = E \setminus \delta^{+}_G(v)$ and $C'=C$. 
\end{defn}

\begin{proposition}
Let $f$ be a delegation rule satisfying guru-participation. Then $f$ also satisfies guru-partition*.
\end{proposition}

\begin{proof}
Let $(G,r)$ be some ranked delegation instance and $v \in D$ with $f(G,r,v)$ ending in casting voter $c$. Let $(G',r)$ be the instance as specified within the definitions of guru-participation and guru-participation*. Let $C^{(0)} = \{c \in C \mid \tau_c = 0\}$ and $C^{(1)} = \{c \in C \mid \tau_c = 1\}$ and assume wlog that $\tau_c = 1$.

Since $f$ satisfies guru-participation and $c \not \in C^{(0)}$, we know that 
\begin{equation}
    \sum_{u \in C^{(0)}} \omega_f(G,r,u) \leq \sum_{u \in C^{(0)}} \omega_f(G',r,u)\text. \label{bound_for_c}
\end{equation}

This leads to 
\begin{align*}
    \sum_{u \in C} \omega_f(G,r,u) \cdot \tau_c & = \sum_{u \in C^{(1)}} \omega_f(G,r,u) \\
    & = 1 - \sum_{u \in C^{(0)}} \omega_f(G,r,u) \\ 
    & \stackrel{\cref{bound_for_c}}{\geq} 1 - \sum_{u \in C^{(0)}} \omega_f(G',r,u) \\
    & = \sum_{u \in C^{(1)}} \omega_f(G',r,u) \\ 
    & = \sum_{u \in C} \omega_f(G,r,u) \cdot \tau_c.
\end{align*}
Hence, either $o = o'$ or $o \succ_c o'$, that is, the election outcome for election $(G,r)$ is weakly preferred by the casting voter $c$ over the election outcome for the election $(G',r)$. 
\end{proof}

For the reverse direction to hold we need to make one additional assumption, namely that the delegation rule is not influenced by isolated casting voters. This is just a technicality required for the proof and any reasonable delegation should satisfy this. 

\begin{defn}
A delegation rule satisfied \emph{independence of isolated casting voters (IIC)}, if for every instance $(G=(C\cup D\cup I),r)$ the following holds: Let $(G'=(C' \cup D \cup I,E),r)$ be derived from $G$ by setting $C' = C \cup \{c^*\}$, where $c^* \not\in C$. Then, $f(G,r,v) = f(G',r,v)$ for all $v \in D$.  
\end{defn}

Before we prove the reverse direction we prove a helpful lemma. 

\begin{lemma}
Let $a,b,c,d,y \in \mathbb{N}$ such that $b > d$ and $\frac{a}{b} > \frac{c}{d}$.
Then, \[\frac{a}{b+y} > \frac{c}{d+y}\text.\] \label{lem:helper}
\end{lemma}

\begin{proof}
From $b >d$ we also know $a > c$. Hence, we get 
\begin{align*}
    \frac{b+y}{a} &= \frac{b}{a} + \frac{y}{a} < \frac{d}{c} + \frac{y}{c}  = \frac{d + y}{c}, 
\end{align*}
which is equivalent to the claim. 
\end{proof}

We are now ready to prove the reverse direction. 

\begin{proposition}
Let $f$ be a delegation function satisfying IIC and guru-participation*. Then, $f$ satisfies guru-participation. 
\end{proposition}

\begin{proof}
We show the contrapositive of the statement. Let $f$ be a delegation rule satisfying IIC but not guru-participation. 

Then, there exists an instance $(G,r)$ and some voter $v \in D$ whose assigned delegation path $f(G,r,v)$ ends in some casting voter $c$ such that $\omega_f(G,r,u) > \omega_f(G',r,u)$ for some $u \in C \setminus \{c\}$, where $G'$ is defined as in the definition of guru-participation.  

We are going to show that this leads to a violation of guru-participation*. 
To this end we create an instance $(\hat{G},r)$ from $G$ such that the outcome of the election will be a tie. For that we will first add some dummy casting voters and then define the sets $C^{0}$ and $C^{(1)}$, denoting the set of casting voters voting for $0$ and $1$ in the final election, respectively.

We denote by $z_f(G,r,u)$ the absolute voting weight of a voter in instance $(G,r)$, i.e. 
\[z_f(G,r,u) = \omega_f(G,r,u) \cdot (|C| + |D|)\text.\]

\textbf{Case 1:} $\omega_f(G,r,u) \leq 0.5$. Then, we create $\hat{G}$ from $G$ by adding $y = |C| + |D| -2z_f(G,r,u)$ ``dummy'' isolated casting voters. Then, let $C^{(0)}$ to be the set consisting of $u$ and all dummy casting voters and $C^{(1)}= \{C \setminus \{u\}\}$. Now, IIC ensures that 
\begin{align*}
\sum_{w \in C^{(0)}} z_f(\hat{G},r,w) &= |C| + |D| - z_f(\hat{G},r,u) \\
&= \sum_{w \in C^{(1)}} z_f(\hat{G},r,w),  
\end{align*}
and hence the outcome of the election is a tie. Finally, let $\hat{G}'$ be the graph obtained by $G'$ by adding the $y$ isolated casting voters, setting $C'{(0)} = C{(0)}$ and $C^{(1)}= C^{(1)}$. 

\textbf{Case 2:} $\omega_f(G,r,u) > 0.5$. We create $\hat{G}$ by adding $y = 2z_f(G,r,u) - |C| - |D|$ ``dummy'' isolated casting voters. Now, $C^{(0)}= \{u\}$ and $C^{(1)}$ consists of all $C \setminus \{u\}$ plus all dummy casting voters. Again, by IIC we get that the outcome of the election $\hat{G}$ is a tie. Finally, let $\hat{G}'$ be the graph obtained by $G'$ by adding the $y$ isolated casting voters, setting $C'{(0)} = C{(0)}$ and $C^{(1)}= C^{(1)}$.

In both cases, we get by IIC that 
\begin{align}
    \omega_f(\hat{G},r,u) & = \frac{z_f(\hat{G},r,u)}{|C| + |D| + y}  \nonumber = \frac{z_f(G,r,u)}{|C| + |D| + y}  \nonumber \\ 
    & > \frac{z_f(G',r,u)}{|C'| + |D'| + y}   = \frac{z_f(\hat{G}',r,u)}{|C'| + |D'| + y}  \label{bound_for_u} \\ 
    &= \omega_f(\hat{G}',r,u) \text, \nonumber
\end{align}
where the inequality follows from \Cref{lem:helper} and $\omega_f(G,r,u) > \omega_f(G',r,u)$. 
We use this to conclude
\begin{align*}
    \sum_{a \in C} \omega_f(\hat{G}',r,a) \cdot \tau_a  & = \sum_{a \in C^{(1)}} \omega_f(\hat{G}',r,a) \\ 
    & = 1 - \sum_{a \in C^{(0)}} \omega_f(\hat{G}',r,a)\\
    & \stackrel{\cref{bound_for_u}}{>}  1 - \sum_{a \in C^{(0)}} \omega_f(\hat{G},r,a) \\ 
    & = \sum_{a \in C^{(1)}} \omega_f(\hat{G},r,a) \\ 
    & = \sum_{a \in C} \omega_f(\hat{G},r,a) \cdot \tau_a \\ 
    & = 0.5\text. 
\end{align*}
Hence, the outcome of $\hat{G}'$ is $1$ while the outcome of $\hat{G}$ is $0.5$. Since $c$ prefers outcome $1$, the instances $\hat{G}$ and $\hat{G}'$ (together with the partition into $C^{(0)}$ and $C^{(1)}$) witness a contradiction of guru-participation*. 
\end{proof}

We are now ready to prove \Cref{confluentGuru}. 

\confluentGuru*

\ifconfluentGuru
\confluentGuruProof
\fi 

Before we give the proof for \cref{prop:bordaGuru}, we prove the following lemma: 

\begin{lemma}
Let $(G,r)$ be a ranked delegation instance and $(\bar{G},\bar{r})$ be the corresponding reduced instance. 
Let $D_1 \dot\cup D_2 = D$ be a partition of the delegating voters. Let $A_1, A_2 \subseteq \bar{E}$ be acyclic and for both $i \in \{1,2\}$ it holds that $\delta^{+}_{A_i}(v) = 1$ for all $v \in D_i$ and $\delta^{+}_{A_j}(v) = 0$ for all $v \in D_j$, $j \in \{1,2\} \setminus \{i\}$. Moreover, it holds that for every $v \in D_2$ there exists a path towards some node in $C$ within $A_2$. Then, $A_1 \cup A_2$ is a $C$-branching in $\bar{G}$. \label{lem:merge-branchings}
\end{lemma}

\begin{proof}
Since by construction all voters in $D$ have outdegree $1$ within $A_1 \cup A_2$, it suffices to show that every delegating voter can reach a casting voter in $A_1 \cup A_2$. We first establish this statement for the voters in $D_2$. By the assumptions within the lemma, node in $D_2$ can reach a casting voter in $C$ via edges in $A_2$. Clearly, the same holds for the set $A_1 \cup A_2$. 

For $v \in D_1$, follow the unique path starting in $v$. Since $A_1$ does not contain any cycle and every node in $D_1$ has outdegree $1$, this path eventually leads either to a casting voter (in which case the claim holds) or to a voter in $D_2$. In this case, we know that there exists a unique path towards a casting voter using only edges in $A_2$. This suffices to conclude the proof. 
\end{proof}

\bordaGuru*
\ifbordaGuruAppendix
\bordaGuruProof
\fi

\copyRobustChar* 
\ifcopyRobustCharAppendix
\copyRobustCharProof
\fi

\impConfCopy*
\ifimpConfCopyAppendix
\impConfCopyProof
\fi

\arbCopy*
\ifarbCopyAppendix
\arbCopyProof
\fi 

\DFDChar*
\ifDFDCharAppendix
\DFDCharProof
\fi 

\BFDChar* 
\ifBFDCharAppendix
\BFDCharProof
\fi

\section{Characterization of Diffusion}

One of the axioms in this characterization is \emph{rank-awareness}, informally speaking implying that a main criterion among which sequences are compared is the maximum rank appearing on a sequence.

\begin{defn}
A sequence rule induced by the order $\fpre$ satisfies \emph{rank-awareness} if, for any two comparable sequences $s, s' \in \sequences$ with $\max(s) < \max(s')$ it holds that $s \fpre s'$.
\end{defn}

Clearly, \diffusion and \lexrank satisfy rank-awareness. However, we show in the following that rank-awareness is rather difficult to be achieved together with other properties. 
We start by giving two observations about rank-awareness. 

\begin{proposition}
No sequence rule can be rank-aware and strongly lexicographic.
\end{proposition}

\begin{proof}
Consider $s = (1,1,5)$ and $s' = (2,2,2)$.
A rank-aware sequence rule $f$ would order $s' \succ_f s$ and strongly lexicographic sequence rule $f$ would order $s \succ_f s'$.
\end{proof}

\begin{proposition}
No sequence rule can be rank-aware and copy-robust. \label{thm:lemma-ra-cr}
\end{proposition}

\begin{proof}
Consider $s=(1)$ and $s'=(2)$. Any rank-aware sequence rule ranks $s \fpre s'$ but $s' \fpre (s,3)$. By \Cref{pr: obsCompIff} this is a contradiction to $\fpre$ satisfying copy-robustness. 
\end{proof}

To give a characterization of \diffusion we need to introduce another property, \emph{truncation}, that allows to determine the order of two sequences by only considering a truncated version of the sequences.

\begin{defn}
A sequence rule induced by the order $\fpre$ satisfies \emph{truncation}, if for any sequences $s, s',t,t' \in \sequences$ such that $s$ and $s'$ are comparable and have no joint prefix, and $x \in \mathbb{N}_{\geq 1}$ it holds that: If $x>\max(t)$, $x>\max(t')$ and $(s,x,t) \fpre (s',x,t')$, then $s \fpre s'$.
\end{defn}

\begin{theorem}
\diffusion is the only sequence rule that is confluent, rank-aware, and satisfies truncation.
\end{theorem}
\begin{proof}
It follows directly from the definition of $\difforder$ and Proposition~\ref{cor: diff-confluent} that \diffusion is confluent and rank-aware. 
To see that \diffusion satisfies truncation consider sequences $s, s',t,t' \in \sequences$ such that $s$ and $s'$ have no joint prefix, and $x \in \mathbb{N}_{\geq 1}$ such that $x>\max(t)$ and $x>\max(t')$. Let $(s,x,t) \fpre (s',x,t')$.

\textbf{Case 1}: $\max((s,x,t)) < \max((s',x,t'))$. Then, because $x>\max(t)$, $x>\max(t')$ it holds that $\max(s) < \max(s')$, which implies $s \difforder s'$. 

\textbf{Case 2:} It holds that $\max((s,x,t)) = \max((s',x,t'))$ and $|\argmax((s,x,t))| < |\argmax ((s',x,t'))|$. Then, if $\max(s) < x$, we know that $\max (s) < \max (s')=x$. On the other hand, if $\max(s) \geq x$, it holds that $\max (s) = \max (s')$ and $|\argmax(s)| < |\argmax (s')|$. In any case, $s \difforder s'$.

\textbf{Case 3:} It holds that $\max((s,x,t)) = \max((s',x,t'))$ and $|\argmax((s,x,t))| = |\argmax((s',x,t'))|$. If $x$ is the unique maximum of both $(s,x,t)$ and $(s',x,t')$, then by (iii) in the definition of $\difforder$ it directly follows that $s \difforder s'$. Otherwise, $\max(s) = \max(s'), |\argmax(s)| = |\argmax(s')|$. Let $\bar{s}$ and $\bar{s}'$, respectively, be the prefix of $(s,x,t)$ and $(s',x,t')$, respectively, before the first occurrence of the maximal rank. By (iii) in the definition of $\difforder$ it follows that $\bar{s}\difforder \bar{s}'$. But since $\bar{s}$ and $\bar{s}'$ are also such prefixes for $s$ and $s'$ respectively, $s \difforder s'$ holds.
Thus $(s,x,t) \fpre (s',x,t')$ implies $s \difforder s'$, which proves that \diffusion satisfies truncation.

We prove the uniqueness of \diffusion satisfying the three axioms by contradiction.
Suppose that there exists a sequence rule $f$ induced by some order $\fpre$ on $\sequences$ that also satisfies the three axioms but that differs from \diffusion. That is, there exists two comparable sequences  $s$ and $s'$ such that $s \difforder s'$ and $s' \fpre s$. Let $s$ and $s'$ be such that $|s|+|s'|$ is minimal.
Because of the minimality of $|s|+|s'|$ and because both $f$ and \diffusion are confluent and Proposition~\ref{pr: prefix-invariant}, $s$ and $s'$ have no common prefix, i.e., $s_1 \neq s'_1$.
Because $s \difforder s'$ and $s' \fpre s$ and both $f$ and \diffusion satisfy rank awareness, $\max (s) = \max (s')$. 
Let $i$ be the highest index such that $s_i = \max(s)$ and $j$ be the highest index such that $s'_j = \max(s')$. 
Suppose $i,j >1$.
Then by truncation, $(s_1, \dots, s_{i-1}) \difforder (s'_1, \dots, s'_{j-1})$ and at the same time $(s'_1, \dots, s'_{j-1}) \fpre (s_1, \dots, s_{i-1})$. This is a contradiction to the minimality of $|s|+|s'|$. 
Now assume that $i=1$, i.e. $s_1 = max(s)$ and because $s_1 \neq s'_1$ we have $j\neq1$.
Because of confluence and \cref{pr: prefix-invariant}, for any $x \in \mathbb{N}_{\geq 1}$ we have
$(x, s) \difforder (x,s')$ and $(x,s') \fpre (x,s)$. 
Consider the particular case of $x = s_1 + 1$. Then by truncation, $(x) \difforder (x,s'_1, \dots, s'_{j-1})$ and at the same time $(x,s'_1, \dots, s'_{j-1}) \fpre (x)$. This is a contradiction to the minimality of $|s|+|s'|$ unless $|s|=1$ and $|s'|=j$. Thus assume  $|s|=1$ and $|s'|=j$. By confluence and \cref{pr: prefix-invariant}, $s = (s_1) \fpre (s'_1, \dots, s'_{j-1}, s_1) = s'$. This is a contradiction to our assumption $s' \fpre s$.
As the case for $j=1$ is analogous to $i=1$, we have thus proven that there cannot be another sequence rule than \diffusion that is confluent, rank-aware and satisfies truncation.
\end{proof}

\section{Extended Version of Section \ref{sec:experiments}}

\par 
Our experimental work is only briefly discussed in the main paper. We use this section of the appendix to provide a more in-depth analysis. For the sake of completeness, we repeat some of the information given in the main text.

\subsection{System setup}
\par
The experiments were carried out on a system with a 1.4 GHz Quad-Core Intel Core i5 CPU, 8GB RAM, and macOS 11.2.3 operating system. The software was implemented in Python 3.8.8 and the libraries matplotlib 3.3.4, numpy 1.20.1, and pandas 1.2.4 were used. Additionally, Gurobi 9.1.2 was used to solve linear programs. 

\subsection{Data generation}
\par 
Initially, all methods choose the set of casting voters $C \subseteq V_H$ by selecting every voter independently with probability $p_c \in [0,1]$. In a second step the set $E \subseteq E_H$ and $r$ are constructed simultaneously.

\paragraph{Friendship-based method}
When the base graph $H = (V_H, E_H)$ is undirected and unweighted, we assume edges to represent ``friendship'' links between the nodes. We quantify the ``strength'' of a friendship between two adjacent nodes $v$ and $w$ by the number of common friends, i.e., $\lambda(v, w) = \left |  \delta_H (v) \cap \delta_H (w) \right|$, where $\delta_H(\cdot)$ is the set of neighbors of a node in $H$. 
Then, for every $v \in V\setminus C$, we insert edges in $\delta_H(v)$ one by one as out-going edges for $v$ in $G$ where the probability of selecting the edge $\{v,w\} \in E_H$ is proportional to $(1 + \lambda(v,w))^{\alpha}$, and $\alpha>0$ is a constant. The ranking $r$ among edges is defined by the order of selection. 

The real-world network we used for this method is a subgraph of the Facebook network \citep{viswanath-2009-activity}. 
For synthetic networks, the input network $H$ is generated by the standard Erd\H{o}s--R\'{e}nyi model $G(n,p)$, where we chose the edge probability $p \in [0,1]$ such that the average number of edges per voter is equal some $\Delta >0$.

\paragraph{Prominence-based method}
\citet{KKH+15a} showed in an empirical study that the liquid democracy network of the German Pirate Party has a ``power-law like'' indegree distribution. Motivated by this, \citet{GKMP18b} proposed a method for generating synthetic liquid democracy data in a similar fashion to preferential attachment models \cite{BaAl99a}. On a high level, the prominence of a voter determines how likely it is for other voters to delegate to them: voters tends to delegate their vote to prominent agent, i.e. those who have a high ingoing degree. We adjust this approach to ranked delegation for base graphs $H$ which are directed and unweighted (e.g. the \emph{Twitter} network). As a reminder, $\delta_H^-(x)$ is the number of ingoing edges of $x$ in $H$.

\begin{itemize}
    \item For real networks, we used the \emph{Slashdot Zoo} graph from \cite{slashdot} and the \emph{Twitter} graph from \cite{De_Domenico_2013} as our base graph $H$. We motivate this choice by the fact that \cite{kumar2006} observe that social networks of this kind often have degree distributions that follow a power law, and hence reflect the idea of our prominence-based method. 
    Then, for every $v \in V\setminus C$, we insert edges in $\delta_H(v)$ one by one as out-going edges for $v$ in $G$ where the probability of selecting the edge $\{v,x\} \in E_H$ is proportional to $(1 + \delta_H^-(x))^{\beta}$, and $\beta>0$ is a constant. The ranking $r$ among edges is defined by the order of selection. 
    \item For synthetic data, we set $H$ to be the complete directed graph. We iteratively choose edges from $E_H$ to include in $G$ as follows: Pick a non-casting voter $v$ uniformly at random and add a new outgoing edge $(v,x) \in E_H \setminus E_G$ to $G$ with probability proportional to $(1+|\delta_G^-(x)|)^\beta$. Observe that, in contrast to the real networks, this probability depends on $\delta_G^{-}(x)$ instead of $\delta_H^{-}(x)$. As $\delta_G^{-}(x)$ changes over time, this leads to a self-reinforcing effect  where prominent voters become even more prominent. The process stops after $m = \Delta \cdot |V \setminus C|$ edges have been added to $G$, where $\Delta >0$.
\end{itemize}

\paragraph{Weight-based method}
When the base graph $H$ is a weighted directed graph, with assume edge weights to reflect trust levels between agents. Therefore, the outgoing edges of some node $x$ in $G$ are ordered in decreasing weight (the top-ranked edge being the one with the highest rank). If there is a tie between several edges, we order uniformly at random between the tied edges. Our real networks are taken from the bitcoin exchange platforms \emph{Bitcoin Alpha} and \emph{Bitcoin OTC} \cite{kumar2016edge}, on which users express how much they trust other users on a scale from $-10$ to $+10$. We set $E_G$ to be the set of all edges in $E_H$ with positive weight, as distrust is not modeled in our liquid democracy setting.  For synthetic networks, we randomly generated $n$ nodes on a 2D plane using either a uniform or a Gaussian distribution. For each node, $E_G$ contains edges to the $\Delta$ closest neighbours.

\paragraph{Networks properties}
Table \ref{tab:network_prop} shows some properties of the various real networks used in our experiments. Note that in ``prominence networks'' (i.e., \emph{Twitter} and \emph{Slashdot}), some users have a lot of ingoing edges and thus a lot of people will delegates to them, while in the ``friendship network'', users will likely delegates to their close friends, which will create a lot of cycles. The table also shows the mean participation rate $1-\frac{|I|}{|V|}$ averaged over the several instances of the experiments when the proportion of casting voters is set to $20\%$ (i.e. $p_c = 0.2$). For very connected networks in which the average outdegree is high (like \emph{Twitter} and \emph{Facebook}), the participation rate is close to $100\%$. 
\begin{table*}[t]
    \centering
    \begin{tabular}{lrrrrr}
    \toprule
         & \emph{Twitter} & \emph{Slashdot} & \emph{Facebook} & \emph{Bitcoin OTC} & \emph{Bitcoin Alpha} \\ 
         \midrule
        nodes & $456,626$ & $79,120$ & $63,731$ & $5,881$ & $3,783$ \\
        edges & $14,855,842$& $515,397$ & $817,035$ & $35,592$ & $24,186$ \\
        average degree & $65$ & $13$ & $26$ &  $12$ & $13$\\
        max outdegree & $1,259$ & $426$ & $1,098$ & $763$ & $490$\\
        max indegree & $51,386$ & $2,529$ &  & $535$ &  $398$\\
        \midrule
        delegation generation method & prominence & prominence & friendship & weight-based & weight-based \\
        \midrule
        mean participation rate for $p_c=0.2$ & $100\%$ & $90\%$ & $100\%$ & $84\%$ & $89\%$ \\
        %diameter & $9$ & $12$& $15$ & $9$ & $10$ \\
    \bottomrule
    \end{tabular}
    \caption{Properties of the real networks used in our experiments.
    }
    \label{tab:network_prop}
\end{table*}
\par 

\subsection{Impact of backup delegation}
\par

In this paper, we argue that the liquid democracy framework would benefit from allowing voters to delegate to several agents. Intuitively, enabling voters to have backup delegation creates more possibilities for delegation paths, and therefore this will lower the number of isolated voters. In order to see how important the impact of backup delegation is, we conducted experiments. For each method, we looked at the average fraction of isolated voters $\frac{|I|}{|V|}$ depending on the maximum outdegree allowed in the delegating graph. We denote this maximum outdegree $d$ and say that we only keep edges $(u,v) \in E$ such that $r((u,v)) \le d$. For $d=0$, only the casting voters take part in the election and $D = \emptyset$. The case $d=1$ corresponds to the classic model of liquid democracy, in which every voter delegate to at most one agent. Moreover, we conducted this experiment with various proportions of casting voters, ranging from $1\%$ to $50\%$ of the total number of voters.
\par

\paragraph{Friendship-based method}
\begin{figure}[!ht]
    \centering
    % This file was created by tikzplotlib v0.9.8.
\begin{tikzpicture}[scale=.8]

\definecolor{color0}{rgb}{0.12156862745098,0.466666666666667,0.705882352941177}
\definecolor{color1}{rgb}{1,0.498039215686275,0.0549019607843137}
\definecolor{color2}{rgb}{0.172549019607843,0.627450980392157,0.172549019607843}
\definecolor{color3}{rgb}{0.83921568627451,0.152941176470588,0.156862745098039}
\definecolor{color4}{rgb}{0.580392156862745,0.403921568627451,0.741176470588235}

\begin{axis}[
height=6cm,
legend cell align={left},
legend style={fill opacity=0.8, draw opacity=1, text opacity=1, draw=white!80!black},
tick align=outside,
tick pos=left,
title={Impact of backup delegation},
width=10cm,
x grid style={white!69.0196078431373!black},
xlabel={Maximum outdegree},
xmajorgrids,
xmin=0, xmax=4,
xtick={0,1,2,3,4},
xticklabels={0,1,2,3,4},
xtick style={color=black},
y grid style={white!69.0196078431373!black},
ylabel={Isolated voters},
ymajorgrids,
ymin=0, ymax=1,
ytick style={color=black},
ytick={0,0.2,0.4,0.6,0.8,1},
yticklabels={0\%, 20\%, 40\%, 60\%, 80\%, 100\%}
]
\path [fill=color0, fill opacity=0.2]
(axis cs:0,0.485025721263955)
--(axis cs:0,0.517708972613596)
--(axis cs:1,0.105985419998095)
--(axis cs:2,0.00585346607566795)
--(axis cs:3,0.00534520112357106)
--(axis cs:4,0.00534520112357106)
--(axis cs:4,0.00151194173357139)
--(axis cs:4,0.00151194173357139)
--(axis cs:3,0.00151194173357139)
--(axis cs:2,0.00177918698555612)
--(axis cs:1,0.0718921310223134)
--(axis cs:0,0.485025721263955)
--cycle;

\path [fill=color1, fill opacity=0.2]
(axis cs:0,0.789967999689005)
--(axis cs:0,0.815539692618687)
--(axis cs:1,0.414897022580919)
--(axis cs:2,0.0108926392418954)
--(axis cs:3,0.00835959743200509)
--(axis cs:4,0.00835959743200509)
--(axis cs:4,0.0032096333372259)
--(axis cs:4,0.0032096333372259)
--(axis cs:3,0.0032096333372259)
--(axis cs:2,0.00375351460425832)
--(axis cs:1,0.319964515880619)
--(axis cs:0,0.789967999689005)
--cycle;

\path [fill=color2, fill opacity=0.2]
(axis cs:0,0.891866747859027)
--(axis cs:0,0.911311823569544)
--(axis cs:1,0.646204959745846)
--(axis cs:2,0.0119916951457941)
--(axis cs:3,0.00921626907637918)
--(axis cs:4,0.00921626907637918)
--(axis cs:4,0.00406944520933583)
--(axis cs:4,0.00406944520933583)
--(axis cs:3,0.00406944520933583)
--(axis cs:2,0.00547259056849247)
--(axis cs:1,0.531259325968439)
--(axis cs:0,0.891866747859027)
--cycle;

\path [fill=color3, fill opacity=0.2]
(axis cs:0,0.942427955802444)
--(axis cs:0,0.958412953288465)
--(axis cs:1,0.817736577146206)
--(axis cs:2,0.0125718575300149)
--(axis cs:3,0.00915283170427905)
--(axis cs:4,0.00915283170427905)
--(axis cs:4,0.00350625920481229)
--(axis cs:4,0.00350625920481229)
--(axis cs:3,0.00350625920481229)
--(axis cs:2,0.00476905156089491)
--(axis cs:1,0.721354331944703)
--(axis cs:0,0.942427955802444)
--cycle;

\path [fill=color4, fill opacity=0.2]
(axis cs:0,0.986872274702745)
--(axis cs:0,0.993410333992907)
--(axis cs:1,0.982917190147696)
--(axis cs:2,0.0133444450012449)
--(axis cs:3,0.00970380302998364)
--(axis cs:4,0.00970380302998364)
--(axis cs:4,0.004100544796104)
--(axis cs:4,0.004100544796104)
--(axis cs:3,0.004100544796104)
--(axis cs:2,0.00559033760745165)
--(axis cs:1,0.907408896808826)
--(axis cs:0,0.986872274702745)
--cycle;

\addplot [semithick, color0, mark=*, mark size=3, mark options={solid}]
table {%
0 0.501367330551147
1 0.0889387130737305
2 0.00381636619567871
3 0.00342857837677002
4 0.00342857837677002
};
\addlegendentry{50\% casting}
\addplot [semithick, color1, mark=*, mark size=3, mark options={solid}]
table {%
0 0.802753925323486
1 0.367430806159973
2 0.00732302665710449
3 0.00578463077545166
4 0.00578463077545166
};
\addlegendentry{20\% casting}
\addplot [semithick, color2, mark=*, mark size=3, mark options={solid}]
table {%
0 0.901589274406433
1 0.588732123374939
2 0.00873219966888428
3 0.00664281845092773
4 0.00664281845092773
};
\addlegendentry{10\% casting}
\addplot [semithick, color3, mark=*, mark size=3, mark options={solid}]
table {%
0 0.950420379638672
1 0.769545435905457
2 0.00867044925689697
3 0.00632953643798828
4 0.00632953643798828
};
\addlegendentry{5\%  casting}
\addplot [semithick, color4, mark=*, mark size=3, mark options={solid}]
table {%
0 0.99014139175415
1 0.945163011550903
2 0.00946736335754395
3 0.00690221786499023
4 0.00690221786499023
};
\addlegendentry{1\% casting}
\end{axis}

\end{tikzpicture}
    \caption{Evolution of the proportion of isolated voters depending on the percentage of casting voter and the maximum outdegree, for the friendship-based model and synthetic data ($n=1000$,  $\Delta=5$, $\alpha = 2$). Averaged over $100$ instances.}
    \label{fig:participation_friends}
\end{figure}
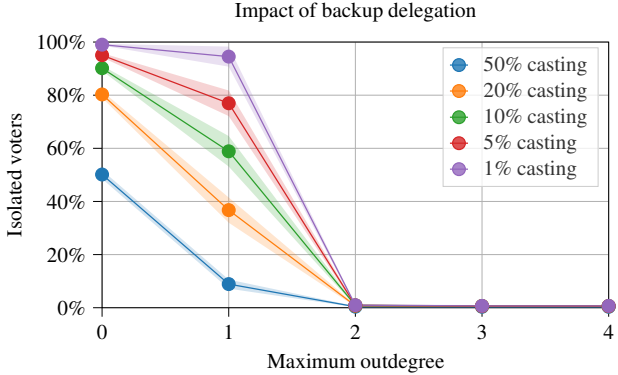

\begin{figure}[!ht]
    \centering
    % This file was created by tikzplotlib v0.9.8.
\begin{tikzpicture}[scale=.8]

\definecolor{color0}{rgb}{0.12156862745098,0.466666666666667,0.705882352941177}
\definecolor{color1}{rgb}{1,0.498039215686275,0.0549019607843137}
\definecolor{color2}{rgb}{0.172549019607843,0.627450980392157,0.172549019607843}
\definecolor{color3}{rgb}{0.83921568627451,0.152941176470588,0.156862745098039}
\definecolor{color4}{rgb}{0.580392156862745,0.403921568627451,0.741176470588235}

\begin{axis}[
height=6cm,
legend cell align={left},
legend style={fill opacity=0.8, draw opacity=1, text opacity=1, draw=white!80!black},
tick align=outside,
tick pos=left,
title={Impact of backup delegation},
width=10cm,
x grid style={white!69.0196078431373!black},
xlabel={Maximum outdegree},
xmajorgrids,
xmin=0, xmax=4,
xtick={0,1,2,3,4},
xticklabels={0,1,2,3,4},
xtick style={color=black},
y grid style={white!69.0196078431373!black},
ylabel={Isolated voters},
ymajorgrids,
ymin=0, ymax=1,
ytick style={color=black},
ytick={0,0.2,0.4,0.6,0.8,1},
yticklabels={0\%, 20\%, 40\%, 60\%, 80\%, 100\%}
]
\path [fill=color0, fill opacity=0.2]
(axis cs:0,0.480954020503746)
--(axis cs:0,0.51057539126096)
--(axis cs:1,0.101331669455204)
--(axis cs:2,0.0140828320247768)
--(axis cs:3,0.00193944486755737)
--(axis cs:4,0.00193944486755737)
--(axis cs:4,-5.70919263811476e-05)
--(axis cs:4,-5.70919263811476e-05)
--(axis cs:3,-5.70919263811476e-05)
--(axis cs:2,0.00368187385757657)
--(axis cs:1,0.0733742128977372)
--(axis cs:0,0.480954020503746)
--cycle;

\path [fill=color1, fill opacity=0.2]
(axis cs:0,0.785066983588245)
--(axis cs:0,0.811019972933494)
--(axis cs:1,0.419867638926745)
--(axis cs:2,0.0658638062547221)
--(axis cs:3,0.00303218329245203)
--(axis cs:4,0.00303218329245203)
--(axis cs:4,0.000533034098852081)
--(axis cs:4,0.000533034098852081)
--(axis cs:3,0.000533034098852081)
--(axis cs:2,0.0337014111365823)
--(axis cs:1,0.323697578464559)
--(axis cs:0,0.785066983588245)
--cycle;

\path [fill=color2, fill opacity=0.2]
(axis cs:0,0.885949366359655)
--(axis cs:0,0.908959724549436)
--(axis cs:1,0.648423120404886)
--(axis cs:2,0.0945071428961256)
--(axis cs:3,0.00394975729127212)
--(axis cs:4,0.00394975729127212)
--(axis cs:4,0.000838121496607114)
--(axis cs:4,0.000838121496607114)
--(axis cs:3,0.000838121496607114)
--(axis cs:2,0.0572504328614499)
--(axis cs:1,0.544485970504205)
--(axis cs:0,0.885949366359655)
--cycle;

\path [fill=color3, fill opacity=0.2]
(axis cs:0,0.943635685248091)
--(axis cs:0,0.955818860206455)
--(axis cs:1,0.815220161329564)
--(axis cs:2,0.115307965915416)
--(axis cs:3,0.00364776732518646)
--(axis cs:4,0.00364776732518646)
--(axis cs:4,0.00114011146269233)
--(axis cs:4,0.00114011146269233)
--(axis cs:3,0.00114011146269233)
--(axis cs:2,0.0831162765088265)
--(axis cs:1,0.726355596246194)
--(axis cs:0,0.943635685248091)
--cycle;

\path [fill=color4, fill opacity=0.2]
(axis cs:0,0.987424391549187)
--(axis cs:0,0.993472160174951)
--(axis cs:1,0.979429826955995)
--(axis cs:2,0.148346485363245)
--(axis cs:3,0.0042407160577056)
--(axis cs:4,0.0042407160577056)
--(axis cs:4,0.000655835666431903)
--(axis cs:4,0.000655835666431903)
--(axis cs:3,0.000655835666431903)
--(axis cs:2,0.0909638594643413)
--(axis cs:1,0.929949483388833)
--(axis cs:0,0.987424391549187)
--cycle;

\addplot [semithick, color0, mark=*, mark size=3, mark options={solid}]
table {%
0 0.49576473236084
1 0.087352991104126
2 0.00888240337371826
3 0.000941157341003418
4 0.000941157341003418
};
\addlegendentry{50\% casting}
\addplot [semithick, color1, mark=*, mark size=3, mark options={solid}]
table {%
0 0.798043489456177
1 0.371782541275024
2 0.0497826337814331
3 0.00178265571594238
4 0.00178265571594238
};
\addlegendentry{20\% casting}
\addplot [semithick, color2, mark=*, mark size=3, mark options={solid}]
table {%
0 0.897454500198364
1 0.596454620361328
2 0.0758787393569946
3 0.00239396095275879
4 0.00239396095275879
};
\addlegendentry{10\% casting}
\addplot [semithick, color3, mark=*, mark size=3, mark options={solid}]
table {%
0 0.949727296829224
1 0.770787954330444
2 0.0992121696472168
3 0.00239396095275879
4 0.00239396095275879
};
\addlegendentry{5\% casting}
\addplot [semithick, color4, mark=*, mark size=3, mark options={solid}]
table {%
0 0.990448236465454
1 0.954689741134644
2 0.119655132293701
3 0.00244832038879395
4 0.00244832038879395
};
\addlegendentry{1\% casting}
\end{axis}

\end{tikzpicture}
    \caption{Evolution of the proportion of isolated voters depending on the percentage of casting voter and the maximum outdegree, for the friendship-based model and synthetic data ($n=1000$,  $\Delta=5$, $\alpha = 10$). Averaged over $100$ instances.}
    \label{fig:participation_friends_hard}
\end{figure}
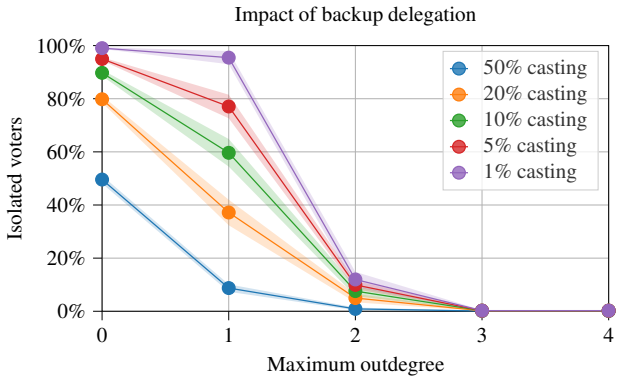

Figure~\ref{fig:participation_friends} shows the results for the friendship-based method with synthetic data. This method often leads to small cycles, because close friends are likely to delegate to each other. Therefore, allowing only one delegation per voter leads to a high proportion of isolated voters, particularly with a small proportion of casting voters. However, with 2 delegations per voters, the proportion of isolated voters immediately drops to almost $0 \%$, even with a small proportion of casting voter. When $\alpha$ is higher, bigger cycles will be created, and $3$ delegations are necessary to reach a proportion of isolated voters close to $0\%$. This is shown on Figure~\ref{fig:participation_friends_hard} in which $\alpha = 10$.

\paragraph{Prominence-based method}

\begin{figure}[!ht]
    \centering
    % This file was created by tikzplotlib v0.9.8.
\begin{tikzpicture}[scale=.8]

\definecolor{color0}{rgb}{0.12156862745098,0.466666666666667,0.705882352941177}
\definecolor{color1}{rgb}{1,0.498039215686275,0.0549019607843137}
\definecolor{color2}{rgb}{0.172549019607843,0.627450980392157,0.172549019607843}
\definecolor{color3}{rgb}{0.83921568627451,0.152941176470588,0.156862745098039}
\definecolor{color4}{rgb}{0.580392156862745,0.403921568627451,0.741176470588235}

\begin{axis}[
height=6cm,
legend cell align={left},
legend style={fill opacity=0.8, draw opacity=1, text opacity=1, draw=white!80!black},
tick align=outside,
tick pos=left,
title={Impact of backup delegation},
width=10cm,
x grid style={white!69.0196078431373!black},
xlabel={Maximum outdegree},
xmajorgrids,
xmin=0, xmax=4,
xtick={0,1,2,3,4},
xticklabels={0,1,2,3,4},
xtick style={color=black},
y grid style={white!69.0196078431373!black},
ylabel={Isolated voters},
ymajorgrids,
ymin=0, ymax=1,
ytick style={color=black},
ytick={0,0.2,0.4,0.6,0.8,1},
yticklabels={0\%, 20\%, 40\%, 60\%, 80\%, 100\%}
]
\path [fill=color0, fill opacity=0.2]
(axis cs:0,0.483405102799958)
--(axis cs:0,0.517694897200042)
--(axis cs:1,0.0289443461339146)
--(axis cs:2,0.012545705427261)
--(axis cs:3,0.012531608320769)
--(axis cs:4,0.012531608320769)
--(axis cs:4,0.00562839167923157)
--(axis cs:4,0.00562839167923157)
--(axis cs:3,0.00562839167923157)
--(axis cs:2,0.00563429457273934)
--(axis cs:1,0.0074956538660863)
--(axis cs:0,0.483405102799958)
--cycle;

\path [fill=color1, fill opacity=0.2]
(axis cs:0,0.787846362277952)
--(axis cs:0,0.811573637722048)
--(axis cs:1,0.12979095151455)
--(axis cs:2,0.019448620058463)
--(axis cs:3,0.0191252300548563)
--(axis cs:4,0.0191252300548563)
--(axis cs:4,0.0109747699451448)
--(axis cs:4,0.0109747699451448)
--(axis cs:3,0.0109747699451448)
--(axis cs:2,0.0109713799415382)
--(axis cs:1,0.0353090484854514)
--(axis cs:0,0.787846362277952)
--cycle;

\path [fill=color2, fill opacity=0.2]
(axis cs:0,0.889421234894988)
--(axis cs:0,0.909178765105012)
--(axis cs:1,0.330964145885855)
--(axis cs:2,0.0220224719791383)
--(axis cs:3,0.0216867213927003)
--(axis cs:4,0.0216805592491328)
--(axis cs:4,0.0127394407508686)
--(axis cs:4,0.0127394407508686)
--(axis cs:3,0.0127532786073012)
--(axis cs:2,0.0128775280208627)
--(axis cs:1,0.079095854114144)
--(axis cs:0,0.889421234894988)
--cycle;

\path [fill=color3, fill opacity=0.2]
(axis cs:0,0.942697945494589)
--(axis cs:0,0.957302054505412)
--(axis cs:1,0.571712266145574)
--(axis cs:2,0.023588428904688)
--(axis cs:3,0.0232078596470657)
--(axis cs:4,0.0232078596470657)
--(axis cs:4,0.0136321403529345)
--(axis cs:4,0.0136321403529345)
--(axis cs:3,0.0136321403529345)
--(axis cs:2,0.0138515710953123)
--(axis cs:1,0.170367733854426)
--(axis cs:0,0.942697945494589)
--cycle;

\path [fill=color4, fill opacity=0.2]
(axis cs:0,0.988383503602709)
--(axis cs:0,0.994812758079534)
--(axis cs:1,1.03737345764937)
--(axis cs:2,0.168379553784015)
--(axis cs:3,0.0239444115246396)
--(axis cs:4,0.0239176291979964)
--(axis cs:4,0.0156898474375186)
--(axis cs:4,0.0156898474375186)
--(axis cs:3,0.0156817566996607)
--(axis cs:2,-0.0918561892980325)
--(axis cs:1,0.669205981602968)
--(axis cs:0,0.988383503602709)
--cycle;

\addplot [semithick, color0, mark=*, mark size=3, mark options={solid}]
table {%
0 0.500550031661987
1 0.0182199478149414
2 0.00908994674682617
3 0.00908005237579346
4 0.00908005237579346
};
\addlegendentry{50\% casting}
\addplot [semithick, color1, mark=*, mark size=3, mark options={solid}]
table {%
0 0.799710035324097
1 0.082550048828125
2 0.0152100324630737
3 0.0150500535964966
4 0.0150500535964966
};
\addlegendentry{20\% casting}
\addplot [semithick, color2, mark=*, mark size=3, mark options={solid}]
table {%
0 0.8992999792099
1 0.205029964447021
2 0.0174499750137329
3 0.0172200202941895
4 0.0172100067138672
};
\addlegendentry{10\% casting}
\addplot [semithick, color3, mark=*, mark size=3, mark options={solid}]
table {%
0 0.950000047683716
1 0.371039986610413
2 0.0187200307846069
3 0.0184199810028076
4 0.0184199810028076
};
\addlegendentry{5\% casting}
\addplot [semithick, color4, mark=*, mark size=3, mark options={solid}]
table {%
0 0.991598129272461
1 0.853289723396301
2 0.0382616519927979
3 0.019813060760498
4 0.0198037624359131
};
\addlegendentry{1\% casting}
\end{axis}

\end{tikzpicture}
    \caption{Evolution of the proportion of isolated voters depending on the percentage of casting voter and the maximum outdegree, for the prominence-based model  ($n=1000$, $\Delta = 4$, $\beta = 2$). Averaged over $100$ instances.}
    \label{fig:app_participation_pop}
\end{figure}
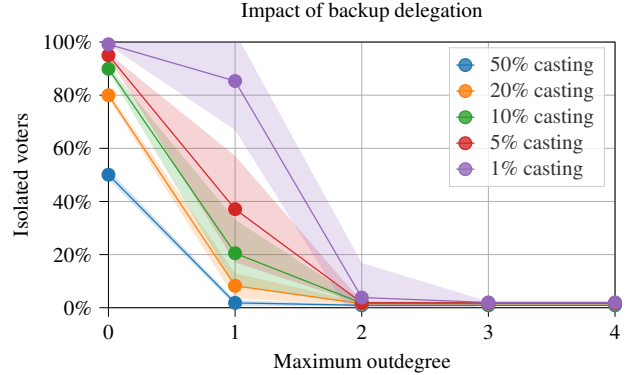

Figure~\ref{fig:app_participation_pop} shows the results for the prominence-based method with synthetic data. Unlike the other methods, the prominence-based method leads to very few cycle, as most voters delegate to the same ''prominent`` agents. Consequently, the fraction of isolated voters when each voter delegates only once is already quite low, except for very small proportion of casting voter. With 2 delegations per voter, the proportion of isolated voters almost drops to almost $0\%$.

\paragraph{Weight-based method}

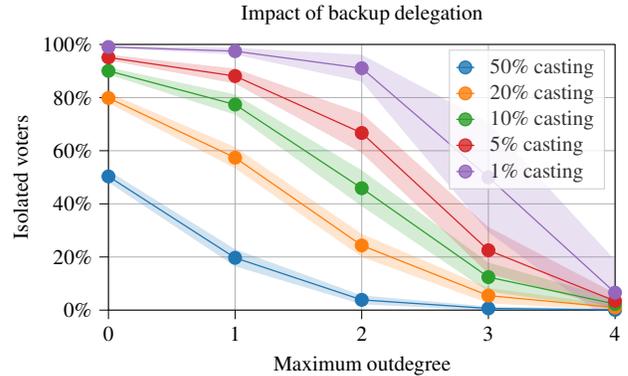
\begin{figure}[!ht]
    \centering
    % This file was created by tikzplotlib v0.9.8.
\begin{tikzpicture}[scale=.8]

\definecolor{color0}{rgb}{0.12156862745098,0.466666666666667,0.705882352941177}
\definecolor{color1}{rgb}{1,0.498039215686275,0.0549019607843137}
\definecolor{color2}{rgb}{0.172549019607843,0.627450980392157,0.172549019607843}
\definecolor{color3}{rgb}{0.83921568627451,0.152941176470588,0.156862745098039}
\definecolor{color4}{rgb}{0.580392156862745,0.403921568627451,0.741176470588235}

\begin{axis}[
height=6cm,
legend cell align={left},
legend style={fill opacity=0.8, draw opacity=1, text opacity=1, draw=white!80!black},
tick align=outside,
tick pos=left,
title={Impact of backup delegation},
width=10cm,
x grid style={white!69.0196078431373!black},
xlabel={Maximum outdegree},
xmajorgrids,
xmin=0, xmax=4,
xtick={0,1,2,3,4},
xticklabels={0,1,2,3,4},
xtick style={color=black},
y grid style={white!69.0196078431373!black},
ylabel={Isolated voters},
ymajorgrids,
ymin=0, ymax=1,
ytick style={color=black},
ytick={0,0.2,0.4,0.6,0.8,1},
yticklabels={0\%, 20\%, 40\%, 60\%, 80\%, 100\%}
]
\path [fill=color0, fill opacity=0.2]
(axis cs:0,0.480916446380971)
--(axis cs:0,0.52564093066821)
--(axis cs:1,0.229071909442441)
--(axis cs:2,0.0554091316450565)
--(axis cs:3,0.0138825250249424)
--(axis cs:4,0.00469151698790138)
--(axis cs:4,-0.00285545141413102)
--(axis cs:4,-0.00285545141413102)
--(axis cs:3,-0.00122678732002424)
--(axis cs:2,0.0218367699942879)
--(axis cs:1,0.164633008590346)
--(axis cs:0,0.480916446380971)
--cycle;

\path [fill=color1, fill opacity=0.2]
(axis cs:0,0.77893163984135)
--(axis cs:0,0.818127183688061)
--(axis cs:1,0.613572374980545)
--(axis cs:2,0.286898800471361)
--(axis cs:3,0.0833162884488704)
--(axis cs:4,0.0204860610864761)
--(axis cs:4,-0.000662531674711309)
--(axis cs:4,-0.000662531674711309)
--(axis cs:3,0.0258601821393655)
--(axis cs:2,0.20010119952864)
--(axis cs:1,0.533957036784161)
--(axis cs:0,0.77893163984135)
--cycle;

\path [fill=color2, fill opacity=0.2]
(axis cs:0,0.886709094720025)
--(axis cs:0,0.914173258221151)
--(axis cs:1,0.81195498412181)
--(axis cs:2,0.526636624879786)
--(axis cs:3,0.178908436015112)
--(axis cs:4,0.0430689986267965)
--(axis cs:4,0.00293100137320379)
--(axis cs:4,0.00293100137320379)
--(axis cs:3,0.0695033286907712)
--(axis cs:2,0.391010433943743)
--(axis cs:1,0.735633251172307)
--(axis cs:0,0.886709094720025)
--cycle;

\path [fill=color3, fill opacity=0.2]
(axis cs:0,0.940177419596578)
--(axis cs:0,0.961795913736755)
--(axis cs:1,0.909171529428243)
--(axis cs:2,0.74160276348177)
--(axis cs:3,0.31274202992191)
--(axis cs:4,0.0634242582217259)
--(axis cs:4,0.00548240844494075)
--(axis cs:4,0.00548240844494075)
--(axis cs:3,0.137604636744756)
--(axis cs:2,0.59183723651823)
--(axis cs:1,0.85280180390509)
--(axis cs:0,0.940177419596578)
--cycle;

\path [fill=color4, fill opacity=0.2]
(axis cs:0,0.986767666569016)
--(axis cs:0,0.994032333430984)
--(axis cs:1,0.987356224229901)
--(axis cs:2,0.961284997142654)
--(axis cs:3,0.706910162593855)
--(axis cs:4,0.188155111113169)
--(axis cs:4,-0.0565858803439387)
--(axis cs:4,-0.0565858803439387)
--(axis cs:3,0.293089837406145)
--(axis cs:2,0.8597611567035)
--(axis cs:1,0.961874545000868)
--(axis cs:0,0.986767666569016)
--cycle;

\addplot [semithick, color0, mark=*, mark size=3, mark options={solid}]
table {%
0 0.503278732299805
1 0.196852445602417
2 0.0386229753494263
3 0.00632786750793457
4 0.000918030738830566
};
\addlegendentry{50\% casting}
\addplot [semithick, color1, mark=*, mark size=3, mark options={solid}]
table {%
0 0.798529386520386
1 0.573764681816101
2 0.243499994277954
3 0.0545881986618042
4 0.00991177558898926
};
\addlegendentry{20\% casting}
\addplot [semithick, color2, mark=*, mark size=3, mark options={solid}]
table {%
0 0.90044116973877
1 0.773794174194336
2 0.458823561668396
3 0.124205827713013
4 0.0230000019073486
};
\addlegendentry{10\% casting}
\addplot [semithick, color3, mark=*, mark size=3, mark options={solid}]
table {%
0 0.950986623764038
1 0.88098669052124
2 0.666719913482666
3 0.225173354148865
4 0.0344533920288086
};
\addlegendentry{5\% casting}
\addplot [semithick, color4, mark=*, mark size=3, mark options={solid}]
table {%
0 0.990400075912476
1 0.974615335464478
2 0.910523056983948
3 0.5
4 0.0657845735549927
};
\addlegendentry{1\% casting}
\end{axis}

\end{tikzpicture}
    \caption{Evolution of the proportion of isolated voters depending on the percentage of casting voter and the maximum outdegree, for the distance-based model ($n=500$, $\Delta = 5$). Averaged over $100$ instances.}
    \label{fig:app_participation_space}
\end{figure}

Figure \ref{fig:app_participation_space} shows the results for the weight-based method with synthetic data. In this method, agents are associated to a position in a 2D space and delegate to their closest neighbour, in increasing distance order. Consequently, it will create a lot of isolated clusters because if $u$ is a very close neighbour of $v$, then $v$ is a very close neighbour of $u$. This is why the proportion of isolated voters remains low even with $2$ delegation per voter, and we need at least $4$ delegations per voter to reach a proportion of isolated voters close to $0 \%$.
\par 
As a conclusion, it is clear that the impact of liquid democracy itself on the proportion of isolated voters in the network is already very high, as it enables people that cannot vote to take part in the election. However, this could be largely improved with backup delegation, and $2$ backup delegation are often sufficient to ensure an proportion of isolated voters close to $0 \%$.

\subsection{Comparison of Delegation Rules}
\paragraph{General comments}
As we already said in the main paper, our experiments show no large differences in the relative performance of the rules for the different methods that generate delegation graphs. For each method, we can separate rules that lead to a more uniform weight distribution (such as \bfd) and rules that ensure low ranks in delegation paths (such as \dfd). Here is a more detailed description of rules performances:
\begin{itemize}
    \item \bfd is by definition the best rule to minimize the length of paths in the $C$-branching. In particular, on every dataset the average lengths of path does not exceed $2$. It also creates $C$-branchings with an almoust uniform weight distribution among casting voters. Apart from these two metrics, \bfd is less convincing: The average rank of a selected edge is between $2$ and $3$, which means that \bfd uses almost always backup delegation. Moreover, \bfd leads to $C$-branching with a large unpopularity margin, making up to $50\%$ of the voters.
    \item \dfd behaves very differently than \bfd. In particular, it leads to a concentrated weight distribution and very long delegation paths, but ranks on delegation paths are usually very low. As \dfd is not a confluent rule, \minsumarb appears to be an interesting confluent alternative to \dfd.
    \minsumarb has similar or better performances on every metric except the weight distribution, in which \minsumarb performs poorly, with one casting voter having up to $30\%$ of all the network weight. Notably, \minsumarb leads to $C$-branchings with very low unpopularity margin, in most cases $0$. This also shows that in most of our generated instances there exists a popular branching (which is not true in general).
    \item \diffusion and \lexrank appear to be a good compromise between providing $C$-branchings  with with low ranks and low unpopularity while selecting shorter paths than \dfd and \minsumarb. \lexrank appears to slightly dominate \diffusion on all metrics. 
    \item Finally, \minsumseq achieves a really good trade-off between path lengths and ranks. In most cases, the average rank does not exceed $2$ and the average length is also close to $2$. However, this rule is close to \bfd, so it leads to $C$-branchings with high unpopularity margin.
\end{itemize}

Note also that some differences between the different networks are due to hyperparameters. For instance, a network with more edges will lead to higher rank for \bfd and longer paths for \dfd. Details are explained in Section \ref{subsec:param}.

\paragraph{Friendship-based method}

Figures \ref{fig:table_friendship_method}-\ref{fig:table_facebook} show the results for the friendship-based method. With this method, the structure of the network already ensure a good weight distribution among casting voters. In particular, the most prominent voter has $1\%$ of the total weight in the Facebook dataset. 

\paragraph{Prominence-based method}

Figures \ref{fig:table_popularity}-\ref{fig:table_twitter} show the results for the prominence-based method. Here, a lot of voters delegate to the same prominent agents. This leads to less cycle, and ''easier`` cases for the rules. In particular, the $C$-branchings returned by \lexrank and \diffusion have low unpopularity margin, and the one returned by \minsumarb have unpopularity of almost $0$. However, this also leads to very concentrated weight distribution, as prominent agents get all the weights.

\paragraph{Weight-based method}

Results of the weight-based method are different for the real networks (Figure \ref{fig:table_btca}-\ref{fig:table_btcotc}) and the synthetic networks (Figure \ref{fig:table_trust}). Indeed, the weights in real networks are based on trust levels between users of bitcoin exchange platforms. Therefore, the results are similar to the prominence-based method as there are prominent users that are very trusted. The weights in the synthetic networks are based on distance between the nodes, which leads to many cycles. For this reason, this lead to similar results than the friendship-based model.

\subsection{Varying the parameters} \label{subsec:param}
\par
We also considered different values for the various parameters of synthetic networks. For every method, changing the parameters did not affect the relative differences between the rules. However, the parameters can still change the absolute values a lot. Here are some interesting observations:

\begin{figure}[!p]
    \centering
    \input{fig_friends2}
    \caption{Friendship-based method with synthetic data, averaged over $1000$ instances ($n=1000$, $p_c = 0.2$, $\Delta = 4$, $\alpha = 2$)}
    \label{fig:table_friendship_method}
\end{figure}

\begin{figure}[!p]
    \centering
    \input{fig_Facebook-reduced}
    \caption{Friendship-based method with Facebook network, averaged over $10$ instances ($p_c = 0.2$, $\alpha = 1$)}
    \label{fig:table_facebook}
\end{figure}

\begin{figure}[!p]
    \centering
    \input{fig_pop1000}
    \caption{Prominence-based method with Synthetic data, averaged over $10$ instances ($n=1000$, $\Delta=4$, $p_c = 0.2$, $\beta = 1$)}
    \label{fig:table_popularity}
\end{figure}

\begin{figure}[!p]
    \centering
    \input{fig_Slashdot}
    \caption{Prominence-based method with Slashdot network, averaged over $10$ instances ($p_c = 0.2$, $\beta = 1$)}
    \label{fig:table_slashdot}
\end{figure}

\begin{figure}[!p]
    \centering
    \input{fig_Higgs}
    \caption{Prominence-based method with Twitter network, averaged over $3$ instances ($p_c = 0.2$, $\beta = 1$)}
    \label{fig:table_twitter}
\end{figure}

\begin{figure}[!p]
    \centering
    \input{fig_space}
    \caption{Weight-based method with spatial network (synthetic data), averaged over $100$ instances ($n=500$, $p_c = 0.1$, $\Delta = 6$)}
    \label{fig:table_trust}
\end{figure}

\begin{figure}[!p]
    \centering
    \input{fig_BitcoinAlpha}
    \caption{Weight-based method with Bitcoin Alpha network, averaged over $10$ instances ($p_c = 0.2$)}
    \label{fig:table_btca}
\end{figure}

\begin{figure}[!p]
    \centering
    \input{fig_BitcoinOTC}
    \caption{Weight-based method with Bitcoin OTC network, averaged over $10$ instances ($p_c = 0.2$)}
    \label{fig:table_btcotc}
\end{figure}

\begin{itemize}
    \item Increasing the number of voters in the network leads to a better division of the weight among casting voters. 
    \item Decreasing the proportion of casting voters leads to a network in which delegating paths are naturally longer and contain higher rank values. This also increases the unpopularity of the $C$-branchings (except for \minsumarb). It is interesting to see that for the prominence-based model, the unpopularity of the $C$-branchings returned by \lexrank and \diffusion is close to $0$ when at least $25\%$ of voters are casting voters.
    \item Increasing the number of outgoing edges per voter $\Delta$ has a different effect for each rule. As it creates more path possibilities, \bfd will pick shorter paths, but with higher ranks, and \dfd longer paths, but with smaller ranks. Finally, as \minsumseq is an compromise between these two rules, the average length and the average rank are both remaining equal when we increase $\Delta$.
    \item Increasing $\beta$ in the prominence-based model leads to a more concentrated weight distribution, as prominent agents become even more prominent. However, $\beta$ has no impact on other metrics. 
    \item Similarly, increasing $\alpha$ in the friendship-based model will lead to slightly longer paths, as it create more cycles, but it has no impact on other metrics.
\end{itemize}

\fi

\end{document}